\documentclass[runningheads]{llncs}
\usepackage[T1]{fontenc}
\usepackage{graphicx}
\usepackage{amsmath,amssymb}
\usepackage{multirow}
\usepackage{float}

\newcommand{\MC}{\operatorname{MC}}
\newcommand{\SC}{\operatorname{SC}}
\newcommand{\OPT}{\operatorname{OPT}}
\newcommand{\med}{\operatorname{med}}
\newcommand{\R}{\mathbb{R}}

\newcommand{\Us}{\scriptscriptstyle\mathrm{U}}
\newcommand{\Ls}{\scriptscriptstyle\mathrm{L}}

\begin{document}

\title{Mechanism Design for Bridge Location with Optional Preferences}
\titlerunning{Mechanism Design for Bridge Location with Optional Preferences}

\author{
Xiaoshuang Geng\inst{1} \and
Wenjing Liu\inst{1,2}\orcidID{0000-0003-4826-2088}%
\thanks{Corresponding author.}\and
Genjie Qin\inst{1} \and
Qizhi Fang\inst{1}
}

\authorrunning{X. Geng et al.}

\institute{
School of Mathematical Sciences, Ocean University of China,
Qingdao, 266100, China
\and
Laboratory of Marine Mathematics, Ocean University of China,
Qingdao, 266100, China\\
\email{gxs15063458497@163.com, \{liuwj,qfang\}@ouc.edu.cn}\\
\email{qgj@stu.ouc.edu.cn}
}

\maketitle

\begin{abstract}
We study the bridge location problem with optional preferences, where two separated regions each contain one prelocated facility. Each agent has a private location and a private preference specifying a nonempty subset of the two facilities in which she is interested. Her individual cost is measured by one of three natural variants: the maximum, the sum, or the minimum of her distances to the facilities in which she is interested. The social planner must design deterministic strategyproof mechanisms that elicit truthful reports, choose the location of a connecting bridge, and approximately minimize either the social cost or the maximum cost.

Our main results are as follows. For the social cost objective, we design optimal mechanisms for the max-variant and sum-variant costs, and provide a 3-approximation mechanism for the min-variant cost, and establish a lower bound of 2 for the min-variant. For the maximum cost objective, we give 5/3-approximation mechanisms for the max-variant and sum-variant costs, and a 3-approximation mechanism for the min-variant cost; we also prove a common lower bound of 5/3 that applies to all three cost variants under the maximum cost objective.
\keywords{Facility location \and mechanism design \and optional preferences \and strategyproofness \and approximation ratio}
\end{abstract}

\section{Introduction}

Facility location is a fundamental optimization problem that involves selecting facility sites under given constraints to achieve a social objective. In many practical settings, agents have private information, such as their locations or preferences over facilities, which is not directly known to the social planner. The planner must therefore make decisions based on the information reported by the agents. Since the outcome affects the agents' individual costs, an agent may have an incentive to strategically misreport her private information. The challenge is to design mechanisms that optimize a social objective and prevent agents from benefiting from misreporting their private information. To address this, Procaccia and Tennenholtz~\cite{procaccia2013amd} introduced approximate mechanism design without money, where the approximation ratio quantifies the performance loss due to strategyproofness. One important class of such problems involves locating a bridge that connects two separated regions.

In many practical settings, two regions are separated by a natural or man-made obstacle, such as a river or a highway. Agents who need to access public facilities on the opposite side face long detours. Building a bridge can improve accessibility, but different bridge locations result in  different travel costs, making site selection a critical optimization problem. Prior work has studied bridge location with prelocated facilities. Qin et al.~\cite{qin2024bridges} considered the case where each region has one facility and every agent must cross the bridge. Qin et al.~\cite{locatingbridge2025} relaxed this assumption by allowing each agent to choose between the facility in her own region and the one on the opposite side. However, both models assume that the two facilities provide identical services and fail to capture scenarios where facilities provide distinct services and agents may be interested in either one of the two facilities or both. 

In this paper, we study the bridge location problem with optional preferences over the two facilities, where each agent's preference specifies the subset of facilities in which she is interested. For agents who are interested in both facilities, their individual cost can be measured in different ways depending on the practical context. We consider three natural variants: under the max-variant, an agent's cost is the maximum of her distances to the facilities in which she is interested; under the sum-variant, it is the total distance; under the min-variant, it is the minimum distance. For each variant, we study both the social cost and the maximum cost, and aim to design strategyproof mechanisms with good approximation ratios.

These three individual cost variants arise naturally in different practical settings. For the max-variant individual cost, suppose that one region has a warehouse for raw materials and the other region has a warehouse for packaging materials. A firm may need both materials to complete an order, or only one of the two materials for production. If the firm needs both materials, then the waiting time is determined by the longer of the two travel times; otherwise, its cost is the distance to the corresponding warehouse. For the sum-variant individual cost, suppose that one region has an elementary school and the other region has a kindergarten. An agent may need to take children to both facilities every morning or to only one of them. If both facilities are needed, her cost is naturally measured by the sum of the two travel distances; otherwise, her cost is the distance to the corresponding facility. For the min-variant individual cost, suppose that each region has a subway station. An agent may be able to reach her destination through both subway stations, or only through one of the two stations. If both stations are available to her, then she will choose the subway station with the shorter travel distance; otherwise, her cost is the distance to the available station.

\subsection{Related Work}

Mechanism design without money for facility location games has been widely studied in recent years. Moulin~\cite{moulin1980strategyproofness} presented a complete characterization of strategyproof mechanisms when agents are located on the line and have single-peaked preferences. Procaccia and Tennenholtz~\cite{procaccia2013amd} first introduced approximate mechanism design without money in the context of facility-location problems. Lu et al.~\cite{lu2009tighter} further studied the two-facility location problem and improved the lower bound for deterministic strategyproof mechanisms under the social cost objective. Lu et al.~\cite{lu2010asymptotically} studied two-facility location games in general metric spaces. These works provide the general background for approximate mechanism design without money.

\par\vspace{1pt}\noindent\textbf{Facility Location Games with Optional Preferences.} Serafino and Ventre~\cite{serafino2016heterogeneous} introduced heterogeneous facility location games with optional preferences, where agents have preferences over facilities and the cost of each agent is the sum of the distances to the facilities she is interested in. Chen et al.~\cite{chen2020optional} extended this model to the continuous-line setting and studied the min-variant and max-variant individual costs. Subsequent work has studied different cost variants and settings, including the min-variant cost~\cite{li2020constant}, the sum-variant cost~\cite{kanellopoulos2023discrete,kanellopoulos2025candidate}, and the max-variant cost with candidate locations~\cite{zhao2023maxvariant,lotfi2024truthful}. Cheng et al.~\cite{cheng2013obnoxious} discussed obnoxious facility location games, where each agent hopes that the facility will be built as far away from her as possible. Kanellopoulos and Voudouris~\cite{kanellopoulos2024obnoxious} later studied obnoxious facilities with optional preferences.

\par\vspace{1pt}\noindent\textbf{Facility Location Games with Prelocated Facilities.} Facility location games with prelocated facilities have been studied in various settings. Chan and Wang~\cite{chan2023accessibility} and subsequent works~\cite{chan2024reducing,chan2025disruptions,chan2024extending} explored the addition of shortcuts and bridges to improve connectivity or access to existing facilities on a line or interval. Qin et al.~\cite{qin2024prelocated} studied locating a new facility alongside a prelocated one on a line. The two works most closely related to ours are the bridge-location models of Qin et al.~\cite{qin2024bridges} and Qin et al.~\cite{locatingbridge2025}, which were briefly introduced above. However, in both models, the facility accessed by an agent is determined uniformly: agents are either required to use the facility on the opposite side~\cite{qin2024bridges} or allowed to use whichever facility is closer ~\cite{locatingbridge2025}. They do not address scenarios where the facilities provide distinct services and agents may be interested in either one of the two facilities or both, which is exactly the optional preference setting we study in this paper.

\subsection{Our Contributions}

In this paper, we study a mechanism design problem for locating a bridge between two separated regions. Each region has one prelocated facility, and agents have optional preferences over the two facilities. Our goal is to design deterministic strategyproof mechanisms that incentivize agents to report truthfully and achieve good approximation ratios for the social objectives. We study three individual cost variants, namely the max-variant, the sum-variant, and the min-variant. As social objectives, we consider the social cost and the maximum cost. We focus on the information structure where both locations and preferences are private. The results are summarized in Table~\ref{tab:summary-results}.

\begin{table}[H]
\caption{A summary of our results.}\label{tab:summary-results}
\centering
\scriptsize
\setlength{\tabcolsep}{2.5pt}
\renewcommand{\arraystretch}{1.12}
\begin{tabular}{|c|c|c|c|c|}
\hline
Social objective & Information structure & Individual cost & Upper bound & Lower bound\\
\hline
\multirow{3}{*}{Social cost} & \multirow{3}{*}{\begin{tabular}{@{}c@{}}Private locations,\\private preferences\end{tabular}} & max-variant & $1$ (Theorem~\ref{appstmt:1}) & $1$\\
 &  & sum-variant & $1$ (Theorem~\ref{appstmt:3}) & $1$\\
 &  & min-variant & $3$ (Theorem~\ref{appstmt:5}) & $2$ (Theorem~\ref{thm:sc-min-both-lb2})\\
\hline
\multirow{3}{*}{Maximum cost} & \multirow{3}{*}{\begin{tabular}{@{}c@{}}Private locations,\\private preferences\end{tabular}} & max-variant & $5/3$ (Theorem~\ref{appstmt:7}) & $5/3$ (Theorem~\ref{appstmt:9})\\
 &  & sum-variant & $5/3$ (Theorem~\ref{appstmt:10}) & $5/3$ (Theorem~\ref{appstmt:9})\\
 &  & min-variant & $3$ (Theorem~\ref{appstmt:12}) & $5/3$ (Theorem~\ref{appstmt:9})\\
\hline
\end{tabular}
\end{table}

\section{Preliminaries}\label{sec:model}

Consider two parallel lines\footnote{This assumption is only used to ensure that the two regions do not intersect in the plane.}, denoted by $L_1$ and $L_2$, which represent two regions separated by an obstacle. There are two facilities $F_1$ and $F_2$, prelocated at $x_F$ on $L_1$ and $y_F$ on $L_2$, respectively. A set of agents $M=\{1,\ldots,m\}$ is located on $L_1$, and a set of agents $N=\{1,\ldots,n\}$ is located on $L_2$. Each agent $i\in M$ has a location $x_i\in\mathbb{R}$ on $L_1$, and each agent $j\in N$ has a location $y_j\in\mathbb{R}$ on $L_2$. Let $\mathbf x=(x_1,\ldots,x_m)\in\mathbb{R}^m$ denote the location profile of the agents in $M$, and let $\mathbf y=(y_1,\ldots,y_n)\in\mathbb{R}^n$ denote the location profile of the agents in $N$. Each agent has optional preferences over the two facilities, meaning that she may be interested in either one facility or both facilities. For each agent $i\in M$, let $\mathbf p_i^{\mathrm U}=(p_{i,1}^{\mathrm U},p_{i,2}^{\mathrm U})\in\{0,1\}^2$ denote her facility preference vector, where $p_{i,1}^{\mathrm U}+p_{i,2}^{\mathrm U}\ge 1$. Here, $p_{i,k}^{\mathrm U}=1$ means that agent $i$ is interested in facility $F_k$ for $k\in\{1,2\}$. Similarly, let $\mathbf p_j^{\mathrm L}$ denote the facility preference vector of agent $j\in N$. Let $\mathbf p^{\mathrm U}=(\mathbf p_1^{\mathrm U},\ldots,\mathbf p_m^{\mathrm U})$ and $\mathbf p^{\mathrm L}=(\mathbf p_1^{\mathrm L},\ldots,\mathbf p_n^{\mathrm L})$ denote the facility preference profiles of agents in $M$ and $N$, respectively. An instance is denoted by $\mathbf c=(\mathbf x,\mathbf y,\mathbf p^{\mathrm U},\mathbf p^{\mathrm L})$.

To connect the two regions, we aim to build a bridge with one endpoint on $L_1$ and the other endpoint on $L_2$. For notational simplicity, we assume that the bridge is perpendicular to both parallel lines, so its location can be described by a single coordinate $s$, with one endpoint at $s$ on $L_1$ and the other at $s$ on $L_2$. We assume that the bridge itself has zero traversal cost.\footnote{This assumption models settings where the connection adds no cost under the chosen travel metric. For example, in a campus with two zones connected by an automated shuttle, if cost measures only walking distance, agents incur costs walking to and from the shuttle stops, but not during the ride.}

\subsection*{Individual costs}

For each bridge position $s$, we define the travel costs of every agent to the two facilities as follows. For each agent $i\in M$, let $d_{i1}^{\Us}(s)=|x_i-x_F|$ denote her local travel cost to $F_1$, and let $d_{i2}^{\Us}(s)=|x_i-s|+|s-y_F|$ denote her crossing travel cost to $F_2$ through the bridge. Similarly, for each agent $j\in N$, let $d_{j1}^{\Ls}(s)=|y_j-s|+|s-x_F|$ and $d_{j2}^{\Ls}(s)=|y_j-y_F|$ denote her travel costs to $F_1$ and $F_2$, respectively. For an instance $\mathbf c$ and a bridge position $s$, we consider the following three variants of individual cost.

\begin{enumerate}
\renewcommand{\labelenumi}{(\arabic{enumi})}
\item \begin{samepage}
\textbf{Max-variant}: each agent's cost is the maximum of her distances to all facilities in which she is interested,
\begin{align*}
C_i^{\Us}((x_i,\mathbf p_i^{\Us}),s)
&=\max_{k:\,p_{i,k}^{\Us}=1} d_{ik}^{\Us}(s),\quad i\in M,\\
C_j^{\Ls}((y_j,\mathbf p_j^{\Ls}),s)
&=\max_{k:\,p_{j,k}^{\Ls}=1} d_{jk}^{\Ls}(s),\quad j\in N.
\end{align*}
\end{samepage}

\item \begin{samepage}
\textbf{Sum-variant}: each agent's cost is the sum of her distances to all facilities in which she is interested,
\begin{align*}
C_i^{\Us}((x_i,\mathbf p_i^{\Us}),s)
&=\sum_{k:\,p_{i,k}^{\Us}=1} d_{ik}^{\Us}(s),\quad i\in M,\\
C_j^{\Ls}((y_j,\mathbf p_j^{\Ls}),s)
&=\sum_{k:\,p_{j,k}^{\Ls}=1} d_{jk}^{\Ls}(s),\quad j\in N.
\end{align*}
\end{samepage}

\item \begin{samepage}
\textbf{Min-variant}: each agent's cost is the minimum of her distances to all facilities in which she is interested,
\begin{align*}
C_i^{\Us}((x_i,\mathbf p_i^{\Us}),s)
&=\min_{k:\,p_{i,k}^{\Us}=1} d_{ik}^{\Us}(s),\quad i\in M,\\
C_j^{\Ls}((y_j,\mathbf p_j^{\Ls}),s)
&=\min_{k:\,p_{j,k}^{\Ls}=1} d_{jk}^{\Ls}(s),\quad j\in N.
\end{align*}
\end{samepage}
\end{enumerate}
\subsection*{Social objectives}

We study two social objective functions: the social cost, which is the total cost of all agents, and the maximum cost, which is the largest individual cost among all agents. Given an instance $\mathbf c$ and a bridge position $s$, the social cost and the maximum cost are defined as follows.
\begin{align*}
\SC(\mathbf c,s)
&=\sum_{i\in M}C_i^{\Us}((x_i,\mathbf p_i^{\Us}),s)
+\sum_{j\in N}C_j^{\Ls}((y_j,\mathbf p_j^{\Ls}),s),\\
\MC(\mathbf c,s)
&=\max\left\{\max_{i\in M}C_i^{\Us}((x_i,\mathbf p_i^{\Us}),s),
\max_{j\in N}C_j^{\Ls}((y_j,\mathbf p_j^{\Ls}),s)\right\}.
\end{align*}
\par\indent A deterministic mechanism is a function $f:\mathbb{R}^m\times\mathbb{R}^n\times\{0,1\}^{2m}\times\{0,1\}^{2n}\to\mathbb{R}$ that maps the reported locations and preferences of all agents to a bridge position. For an instance $\mathbf c=(\mathbf x,\mathbf y,\mathbf p^{\Us},\mathbf p^{\Ls})$, the bridge position output by the mechanism is denoted by $f(\mathbf c)$.

In this paper, we focus primarily on mechanism design in the setting where both locations and preferences are private. We restrict each agent's location report to her own line. Under this information structure, a mechanism is strategyproof if no agent can decrease her cost by misreporting her private information, regardless of the reports of the other agents. Formally, for any instance $\mathbf c=(\mathbf x,\mathbf y,\mathbf p^{\Us},\mathbf p^{\Ls})$, the following two conditions hold:
\begin{enumerate}
\item $\forall i\in M$, $\forall x_i'\in\R$, and $\forall (\mathbf p_i^{\Us})'\in\{0,1\}^2$ with $(p_{i,1}^{\Us})'+(p_{i,2}^{\Us})'\ge1$,
\begin{align*}
C_i^{\Us}((x_i,\mathbf p_i^{\Us}),f(\mathbf{x},\mathbf{y},\mathbf{p}^{\Us},\mathbf{p}^{\Ls}))
&\le
C_i^{\Us}((x_i,\mathbf p_i^{\Us}),f((x_i',\mathbf{x}_{-i}),\mathbf{y},((\mathbf p_i^{\Us})',\mathbf{p}_{-i}^{\Us}),\mathbf{p}^{\Ls})),
\end{align*}
where $\mathbf{x}_{-i}$ and $\mathbf{p}_{-i}^{\Us}$ denote the location profile and preference profile of agents in $M\setminus\{i\}$;
\item $\forall j\in N$, $\forall y_j'\in\R$, and $\forall (\mathbf p_j^{\Ls})'\in\{0,1\}^2$ with $(p_{j,1}^{\Ls})'+(p_{j,2}^{\Ls})'\ge1$,
\begin{align*}
C_j^{\Ls}((y_j,\mathbf p_j^{\Ls}),f(\mathbf{x},\mathbf{y},\mathbf{p}^{\Us},\mathbf{p}^{\Ls}))
&\le
C_j^{\Ls}((y_j,\mathbf p_j^{\Ls}),f(\mathbf{x},(y_j',\mathbf{y}_{-j}),\mathbf{p}^{\Us},((\mathbf p_j^{\Ls})',\mathbf{p}_{-j}^{\Ls}))),
\end{align*}
where $\mathbf{y}_{-j}$ and $\mathbf{p}_{-j}^{\Ls}$ denote the location profile and preference profile of agents in $N\setminus\{j\}$.
\end{enumerate}
\begingroup\emergencystretch=1.2em\tolerance=2000
\par\indent For a fixed instance $\mathbf c$, we use the following simplified notation whenever there is no ambiguity. We write $C_i^{\Us}(s)$, $C_j^{\Ls}(s)$, $\SC(s)$, and $\MC(s)$ for $C_i^{\Us}((x_i,\mathbf p_i^{\Us}),s)$, $C_j^{\Ls}((y_j,\mathbf p_j^{\Ls}),s)$, $\SC(\mathbf c,s)$, and $\MC(\mathbf c,s)$, respectively. Let $\OPT_{\SC}(\mathbf c)=\min_{s\in\R}\SC(\mathbf c,s)$ and $\OPT_{\MC}(\mathbf c)=\min_{s\in\R}\MC(\mathbf c,s)$ denote the optimal social cost and the optimal maximum cost, respectively.

\par\endgroup
\par\indent A deterministic strategyproof mechanism $f$ is said to have an approximation ratio of at most $\gamma \ge 1$ under the social cost objective if, for every instance $\mathbf{c}$,
$\mathrm{SC}(\mathbf{c},f(\mathbf{c})) \le \gamma\,\mathrm{OPT}_{\mathrm{SC}}(\mathbf{c})$.
The approximation ratio under the maximum cost objective is defined analogously.

The case $x_F=y_F$ is trivial. Hence, we assume $x_F\ne y_F$. Since the two lines share a common real coordinate system, all coordinates can be translated and rescaled simultaneously. Without loss of generality, we assume $x_F=1$ and $y_F=0$ throughout the remainder of this paper.

According to the agents' preferences, we divide them into the following sets:
\[
\begin{aligned}
M_1 &= \{i\in M:\mathbf p_i^{\Us}=(1,0)\}, &
N_1 &= \{j\in N:\mathbf p_j^{\Ls}=(1,0)\},\\
M_2 &= \{i\in M:\mathbf p_i^{\Us}=(0,1)\}, &
N_2 &= \{j\in N:\mathbf p_j^{\Ls}=(0,1)\},\\
M_{12} &= \{i\in M:\mathbf p_i^{\Us}=(1,1)\}, &
N_{12} &= \{j\in N:\mathbf p_j^{\Ls}=(1,1)\}.
\end{aligned}
\]

\begin{lemma}[Optimal bridge interval]\label{lem:bridge-interval}
For each of the max-variant, sum-variant, and min-variant individual costs and under both the social cost and maximum cost objectives, there exists an optimal bridge position $s^*\in[0,1]$.
\end{lemma}
\textbf{Due to space constraints, all proofs are deferred to Appendix A.}

By Lemma~\ref{lem:bridge-interval}, it suffices in the subsequent analysis to consider mechanisms
whose outputs lie in $[0,1]$.

\begin{remark}\label{rem:outside-constant}

For any bridge position $s\in[0,1]$, the costs of all agents in $M_1\cup N_2$, all agents $i\in M\setminus M_1$ with $x_i\geq 1$, and all agents $j\in N\setminus N_2$ with $y_j\leq 0$ are independent of $s$. Indeed, local travel costs are independent of $s$, while $d_{i2}^U(s)=x_i$ for $x_i\geq 1$ and $d_{j1}^L(s)=1-y_j$ for $y_j\leq 0$. We call these agents \textbf{bridge-independent agents}. Since all mechanisms presented below output a point in $[0,1]$, these agents cannot reduce their true costs by misreporting.

For an agent whose true cost depends on the bridge position, reporting a type in one of these classes does not change her true cost function. Such a report may remove her from a relevant set or leave her in that set. We address both cases explicitly in the strategyproofness proofs below, evaluating deviations at the agent's true location and preference.
\end{remark}

\section{Social Cost}

In this section, we study the design of strategyproof mechanisms under the social cost objective. We mainly consider the information structure where both locations and preferences are private. Under this information structure, we analyze three individual cost variants: the max-variant, the sum-variant, and the min-variant.

Our mechanisms are inspired by the social cost algorithm of Qin et al.~\cite{qin2024bridges}. In their model, as the bridge moves from $0$ to $1$, the social cost can be viewed as a sum of piecewise linear functions whose slopes change at agents' locations. The optimal bridge position is obtained by balancing the increasing and decreasing parts of the social cost. In our setting, the same analytical approach remains applicable, but the positions at which their cost functions change may be different from their own locations. This difference arises from the agents' optional preferences and the three individual cost variants considered in our model. Accordingly, our analysis focuses on agents in $M_2 \cup M_{12} \cup N_1 \cup N_{12}$, whose costs may vary with the bridge position. Therefore, we introduce breakpoints to describe these positions. Based on these breakpoints, we design mechanisms for the three individual cost variants.

\subsection{Max-Variant}

Let $A=M_2\cup M_{12}$ and $B=N_1\cup N_{12}$. For any $z\in\mathbb{R}$, let $z^+=\max\{z,0\}$. For each agent $i\in A$, define her breakpoint as
\[
a_i=
\begin{cases}
\med\{0,x_i,1\}, & i\in M_2,\\
\med\{1/2,x_i,1\}, & i\in M_{12}.
\end{cases}
\]
For each agent $j\in B$, define her breakpoint as
\[
b_j=
\begin{cases}
\med\{0,y_j,1\}, & j\in N_1,\\
\med\{0,y_j,1/2\}, & j\in N_{12}.
\end{cases}
\]

\begin{lemma}\label{lem:sc-max-expansion}
Under the max-variant individual cost, for every $i\in A$ and every $j\in B$, there exist $\alpha_i$ and $\beta_j$, independent of $s$, such that, for every bridge position $s\in[0,1]$, $C_i^{\Us}(s)=\alpha_i+2(s-a_i)^+$ and $C_j^{\Ls}(s)=\beta_j+2(b_j-s)^+$.
\end{lemma}

Lemma~\ref{lem:sc-max-expansion} characterizes the individual cost functions in terms of the above breakpoints. For every agent $i\in A$, her cost remains constant until the bridge reaches $a_i$ and then increases at rate $2$. Symmetrically, for every agent $j\in B$, her cost decreases at rate $2$ until the bridge reaches $b_j$ and remains constant afterward. Based on this breakpoint structure, we define the following mechanism.

\par\smallskip
\noindent\textbf{Mechanism 1.} Given an instance $\mathbf{c}$, sort the multiset
$K=\{0\}\cup\{a_i:i\in A\}\cup\{b_j:j\in B\}$,
with all multiplicities retained, in nondecreasing order as
$e_1\leq\cdots\leq e_{|A|+|B|+1}$, and output $s=e_{|B|+1}$.

\begin{theorem}
Mechanism 1 is strategyproof and optimal for the social cost objective under the max-variant individual cost.
\label{appstmt:1}\label{appstmt:2}
\end{theorem}

\subsection{Sum-Variant}

For each agent $i\in A$, define her breakpoint as $a_i=\med\{0,x_i,1\}$. For each agent $j\in B$, define her breakpoint as $b_j=\med\{0,y_j,1\}$.

\begin{lemma}\label{lem:sc-sum-expansion}
Under the sum-variant individual cost, for every $i\in A$ and every $j\in B$, there exist $\alpha_i$ and $\beta_j$, independent of $s$, such that, for every bridge position $s\in[0,1]$, $C_i^{\Us}(s)=\alpha_i+2(s-a_i)^+$ and $C_j^{\Ls}(s)=\beta_j+2(b_j-s)^+$.
\end{lemma}

\par\smallskip
\noindent\textbf{Mechanism 2.} Given an instance $\mathbf{c}$, sort the multiset
$K=\{0\}\cup\{a_i:i\in A\}\cup\{b_j:j\in B\}$,
with all multiplicities retained, in nondecreasing order as
$e_1\leq\cdots\leq e_{|A|+|B|+1}$, and output $s=e_{|B|+1}$.

\begin{theorem}
Mechanism 2 is strategyproof and optimal for the social cost objective under the sum-variant individual cost.
\label{appstmt:3}\label{appstmt:4}
\end{theorem}

\subsection{Min-Variant}

Under the min-variant individual cost, if an agent $i\in M_{12}$ satisfies $x_i\ge 1/2$, then her nearest interested facility is always $F_1$, so her cost is independent of the bridge position. Similarly, agents $j\in N_{12}$ with $y_j\le 1/2$ have costs independent of the bridge position. Therefore, the mechanism only needs to consider the following sets of agents:
\[
A_{\min}=M_2\cup\{i\in M_{12}:x_i<1/2\},\qquad
B_{\min}=N_1\cup\{j\in N_{12}:y_j>1/2\}.
\]
For each agent $i\in A_{\min}$, define her breakpoint as
\[
a_i=
\begin{cases}
\med\{0,x_i,1\}, & i\in M_2,\\
\med\{0,x_i,1/2\}, & i\in M_{12}\text{ and }x_i<1/2.
\end{cases}
\]
For each agent $j\in B_{\min}$, define her breakpoint as
\[
b_j=
\begin{cases}
\med\{0,y_j,1\}, & j\in N_1,\\
\med\{1/2,y_j,1\}, & j\in N_{12}\text{ and }y_j>1/2.
\end{cases}
\]

\begin{lemma}\label{lem:sc-min-auxiliary}
Under the min-variant individual cost, for every $i\in A_{\min}$ and every $j\in B_{\min}$, there exist $\alpha_i$ and $\beta_j$, independent of $s$, such that, for every bridge position $s\in[0,1]$, $d_{i2}^{\Us}(s)=\alpha_i+2(s-a_i)^+$ and $d_{j1}^{\Ls}(s)=\beta_j+2(b_j-s)^+$.
\end{lemma}

Under the min-variant, the complete individual costs do not admit the same one-sided breakpoint representation as in Lemmas~\ref{lem:sc-max-expansion} and~\ref{lem:sc-sum-expansion}. Therefore, Lemma~\ref{lem:sc-min-auxiliary} instead characterizes the crossing travel costs. The sum of these crossing travel costs over the agents in $A_{\min}$ and $B_{\min}$ forms an auxiliary objective minimized by the following mechanism.

\par\smallskip
\noindent\textbf{Mechanism 3.} Given an instance $\mathbf{c}$, sort the multiset
$K_{\min}=\{0\}\cup\{a_i:i\in A_{\min}\}
\cup\{b_j:j\in B_{\min}\}$,
with all multiplicities retained, in nondecreasing order as
$e_1\leq\cdots\leq e_{|A_{\min}|+|B_{\min}|+1}$,
and output $s=e_{|B_{\min}|+1}$.

\begin{theorem}
Mechanism 3 is strategyproof and a $3$-approximation mechanism for the social cost objective under the min-variant individual cost.
\label{appstmt:5}\label{appstmt:6}
\end{theorem}

\begin{theorem}\label{thm:sc-min-both-lb2}
No deterministic strategyproof mechanism has an approximation ratio less than $2$ for the social cost objective under the min-variant individual cost.
\end{theorem}

The setting of Qin et al.~[16] is a special case of our min-variant in which every agent is interested in both facilities. Any deterministic strategyproof mechanism for our model remains strategyproof on this restricted domain with the same approximation guarantee. Therefore, their lower bound of $2$ for the social cost also applies to our model. We provide an independent proof in Appendix~A based solely on preference misreports.

\section{Maximum Cost}

In this section, we study the design of strategyproof mechanisms under the maximum cost objective. We focus on the information structure where both locations and preferences are private. Under this information structure, we consider three individual cost variants: the max-variant, the sum-variant, and the min-variant. Throughout this section, let $A=M_2\cup M_{12}$ and $B=N_1\cup N_{12}$, as in Section~3.1.

We first give a unified lower bound for the maximum cost objective. When every agent is interested only in the facility in the opposite region, all three individual cost variants coincide, and our model reduces to that of Qin et al.~[18]. Therefore, their lower bound of $5/3$ applies to all three variants. For completeness, we provide a self-contained proof in Appendix~A.

\begin{theorem}
No deterministic strategyproof mechanism has an approximation ratio less than $5/3$ for the maximum cost objective under any of the three individual cost variants.
\label{appstmt:9}
\end{theorem}

Under the maximum cost objective, a good bridge position should balance the crossing travel costs of agents from the two regions. We therefore focus on the minimum location among the agents in $A$ and the maximum location among the agents in $B$. If $A\ne\varnothing$, let $a=\min_{i\in A}x_i$; if $B\ne\varnothing$, let $b=\max_{j\in B}y_j$. Based on these two locations, we define the following mechanism, which is used for all three individual cost variants.

\par\smallskip
\noindent\textbf{Mechanism 4.} Given an instance $\mathbf c$, output
\[
s=
\begin{cases}
\operatorname{med}\left\{0,\tfrac{1}{2},1,a,b\right\},
& A\neq\emptyset,\ B\neq\emptyset,\\
0,
& B=\emptyset,\\
1,
& A=\emptyset,\ B\neq\emptyset.
\end{cases}
\]

\subsection{Max-Variant}

\begin{theorem}
Mechanism 4 is strategyproof and a $5/3$-approximation mechanism for the maximum cost objective under the max-variant individual cost.
\label{appstmt:7}\label{appstmt:8}
\end{theorem}

\subsection{Sum-Variant}

\begin{theorem}
Mechanism 4 is strategyproof and a $5/3$-approximation mechanism for the maximum cost objective under the sum-variant individual cost.
\label{appstmt:10}\label{appstmt:11}
\end{theorem}

\subsection{Min-Variant}

To prove the approximation ratio, we first give two lower bounds on the optimal maximum cost based on the output of Mechanism 4.

\begin{lemma}\label{lem:minmax-control-lb}
Given an instance $\mathbf c$, let the output of Mechanism 4 be $s_0\in[0,1]$. If there exists $i\in A$ with $x_i<s_0$, then $\OPT_{\MC}(\mathbf c)\ge s_0$. If there exists $j\in B$ with $y_j>s_0$, then $\OPT_{\MC}(\mathbf c)\ge 1-s_0$.
\end{lemma}

\begin{theorem}
Mechanism 4 is strategyproof and a $3$-approximation mechanism for the maximum cost objective under the min-variant individual cost.
\label{appstmt:12}\label{appstmt:13}
\end{theorem}

\section{Conclusion and Future Work}

This paper studied deterministic strategyproof mechanisms for locating a bridge between two separated regions with optional preferences, where both agents' locations and preferences are private. For social cost minimization, we introduced breakpoints and designed optimal mechanisms for the max-variant and sum-variant individual costs. For the min-variant individual cost, a gap remains between the approximation ratio of $3$ achieved by our mechanism and the lower bound of $2$. For maximum cost minimization, we designed mechanisms that achieve the tight approximation ratio of $5/3$ for the max-variant and sum-variant individual costs, while the min-variant individual cost retains a gap between the upper bound of $3$ and the lower bound of $5/3$. With a fixed public traversal cost $\tau>0$, Mechanism~2
remains strategyproof and socially optimal under the
sum-variant, whereas Mechanism~1 is no longer strategyproof
under the max-variant; under the min-variant, the nearest
interested facility may change.
See the discussion on positive traversal costs in Appendix~A
for details.

For future directions, an immediate question is whether the remaining gaps for the min-variant individual cost under the two social objectives can be closed. Under the maximum cost objective, the tight approximation ratios for the max-variant and sum-variant individual costs remain above $1$ in the fully private setting, motivating the study of better approximation guarantees under other information structures. We have obtained preliminary results in this direction. It would also be interesting to study randomized mechanisms and more general preference structures, such as fractional preferences over the two facilities. Finally, extending our results to positive traversal costs
in the remaining settings is another direction for future work.

\bibliographystyle{splncs04}
\bibliography{main}

\clearpage
\appendix
\section{Appendix}

\subsection*{Proof of Lemma~\ref{lem:bridge-interval}.}
\begin{proof}
Recall that we have normalized the facility locations to $x_F=1$ and $y_F=0$. Consider any bridge position $s\notin[0,1]$. We show that it can be moved to a point in $[0,1]$ without increasing any agent's individual cost.

If $s<0$, then for any agent $i\in M$ who is interested in $F_2$, $d_{i2}^{\Us}(0)=|x_i|\le |x_i-s|+|s|=d_{i2}^{\Us}(s)$. For any agent $j\in N$ who is interested in $F_1$, $d_{j1}^{\Ls}(0)=|y_j|+1\le |y_j-s|+|s|+1=|y_j-s|+1-s=d_{j1}^{\Ls}(s)$. Hence, moving the bridge from any position $s<0$ to $0$ does not increase any crossing travel cost, while all local travel costs remain unchanged.

The case $s>1$ is analogous. Moving the bridge from $s$ to $1$ does not increase any crossing travel cost, while all local travel costs remain unchanged. Therefore, no agent's individual cost increases under any of the three individual cost variants. Consequently, neither the social cost nor the maximum cost increases. It follows that there exists an optimal bridge position in $[0,1]$.
\qed
\end{proof}

\subsection*{Proof of Lemma~\ref{lem:sc-max-expansion}.}
\begin{proof}
Consider agents in $A$. If $i\in M_2$, then $C_i^{\Us}(s)=|x_i-s|+s$. If $x_i<0$, then $a_i=0$, so $C_i^{\Us}(s)=-x_i+2s=(-x_i)+2(s-a_i)^+$. If $0\le x_i<1$, then $a_i=x_i$. When $s\le x_i$, we have $(s-a_i)^+=0$, so $C_i^{\Us}(s)=x_i=x_i+2(s-a_i)^+$. When $s>x_i$, we have $C_i^{\Us}(s)=2s-x_i=x_i+2(s-a_i)^+$. If $x_i\ge1$, then $a_i=1$. Since $s\in[0,1]$, we have $(s-a_i)^+=0$, so $C_i^{\Us}(s)=x_i=x_i+2(s-a_i)^+$. Therefore, for every $i\in M_2$, there exists a constant $\alpha_i$ independent of $s$ such that $C_i^{\Us}(s)=\alpha_i+2(s-a_i)^+$.

If $i\in M_{12}$, then $C_i^{\Us}(s)=\max\{|x_i-1|,|x_i-s|+s\}$. If $x_i<1/2$, then $a_i=1/2$. When $s\le1/2$, we have $(s-a_i)^+=0$ and $|x_i-s|+s\le1-x_i$, so $C_i^{\Us}(s)=1-x_i=(1-x_i)+2(s-a_i)^+$. When $s>1/2$, we have $|x_i-s|+s=2s-x_i>1-x_i$, so $C_i^{\Us}(s)=2s-x_i=(1-x_i)+2(s-a_i)^+$. If $1/2\le x_i<1$, then $a_i=x_i$. When $s\le x_i$, we have $(s-a_i)^+=0$ and $|x_i-s|+s=x_i\ge1-x_i$, so $C_i^{\Us}(s)=x_i=x_i+2(s-a_i)^+$. When $s>x_i$, we have $|x_i-s|+s=2s-x_i\ge x_i\ge1-x_i$, so $C_i^{\Us}(s)=2s-x_i=x_i+2(s-a_i)^+$. If $x_i\ge1$, then $a_i=1$. Since $s\in[0,1]$, we have $(s-a_i)^+=0$ and $|x_i-s|+s=x_i\ge x_i-1$, so $C_i^{\Us}(s)=x_i=x_i+2(s-a_i)^+$. Therefore, for every $i\in M_{12}$, there exists a constant $\alpha_i$ independent of $s$ such that $C_i^{\Us}(s)=\alpha_i+2(s-a_i)^+$.

The analysis for agents in $B$ is similar.
\qed
\end{proof}

\subsection*{Proof of Theorem~\ref{appstmt:1}.}
\begin{proof}
We first prove strategyproofness. By Remark~\ref{rem:outside-constant}, regardless of the mechanism's output, bridge-independent agents cannot benefit from misreporting. We show that regardless of the reports of all other agents, no agent $i\in A$ satisfying $x_i<1$ can benefit from misreporting. The analysis for any agent $j\in B$ satisfying $y_j>0$ is similar. For any such agent $i$, let $s_0$ be the output of Mechanism 1 when she reports truthfully, and let $s'$ be the output after any misreport by $i$, while all other reports remain unchanged. Here, $A$ and $B$ are the sets determined when agent $i$ reports truthfully. All cost comparisons are based on the agent's true location and preference. Since agent $i$ is on $L_1$, her misreport does not affect $B$.

If $|B|=0$, then Mechanism 1 outputs $s_0=0$. Since $B$ remains unchanged after any misreport by agent $i$, the mechanism also outputs $s'=0$. In this case, agent $i$ attains her minimum cost on $[0,1]$ at $s_0=0$. Therefore, agent $i$ cannot benefit from misreporting.

Now suppose $|B|>0$. When agent $i$ reports truthfully, Mechanism~1 outputs
$s_0=e_{|B|+1}$. By Lemma~\ref{lem:sc-max-expansion}, $C_i^{\Us}(s)$ reaches its minimum on $[0,1]$ when $s\le a_i$. Thus, if $s_0\le a_i$, the agent cannot benefit. It remains to consider the case $s_0>a_i$.

If the agent's reported preference still contains $F_2$ after the misreport, then her report still contributes one breakpoint to the breakpoint multiset. The misreport only changes her breakpoint from $a_i$ to some reported breakpoint $a_i'\in[0,1]$. This includes reports with $x_i'\geq 1$, which correspond to bridge-independent types but still contribute the breakpoint $a_i'=1$. If $a_i'\leq s_0$, then, in the resulting multiset,
there are still at most $|B|$ elements strictly smaller
than $s_0$ and at least $|B|+1$ elements no greater
than $s_0$. Hence, the $(|B|+1)$-st element remains
$s_0$, so $s'=s_0$. If $a_i'>s_0$, then the new
$(|B|+1)$-st element cannot be smaller than $s_0$,
so $s'\geq s_0$. Since $a_i<s_0\le s'$, Lemma~\ref{lem:sc-max-expansion} gives $C_i^{\Us}(s')\ge C_i^{\Us}(s_0)$. Hence, the agent cannot benefit.

If the agent's reported preference does not contain $F_2$ after the misreport, then she is reported as belonging to $M_1$. Her breakpoint $a_i<s_0=e_{|B|+1}$ is removed,
while the index $|B|+1$ does not change. Hence, the new output satisfies $s'\ge s_0$. Since $s_0>a_i$, Lemma~\ref{lem:sc-max-expansion} gives $C_i^{\Us}(s')\ge C_i^{\Us}(s_0)$. Hence, the agent cannot benefit.

We next prove optimality. For every $s\in[0,1]$, the social cost is
\[
\textstyle
\SC(s)=\sum\limits_{i\in M_1} C_i^{\Us}(s)+\sum\limits_{j\in N_2} C_j^{\Ls}(s)+\sum\limits_{i\in A}\bigl(\alpha_i+2(s-a_i)^+\bigr)+\sum\limits_{j\in B}\bigl(\beta_j+2(b_j-s)^+\bigr).
\]
Let $\SC_0=\sum_{i\in M_1} C_i^{\Us}(s)+\sum_{j\in N_2} C_j^{\Ls}(s)+\sum_{i\in A}\alpha_i+\sum_{j\in B}\beta_j$, which is independent of $s$, and let $\Phi(s)=\sum_{i\in A}(s-a_i)^++\sum_{j\in B}(b_j-s)^+$. Then $\SC(s)=\SC_0+2\Phi(s)$, so minimizing $\SC(s)$ is equivalent to minimizing $\Phi(s)$. Let $s_0$ be the output of Mechanism 1. By Lemma~\ref{lem:bridge-interval}, it suffices to prove that $\SC(s_0)\le\SC(s)$ for every $s\in[0,1]$.

If $|B|=0$, then $\Phi(s)=\sum_{i\in A}(s-a_i)^+$, which is nondecreasing on $[0,1]$. Mechanism 1 outputs $s_0=0$, and hence $s_0$ minimizes $\Phi$.

Suppose $|B|>0$. Mechanism~1 outputs $s_0=e_{|B|+1}$. Since every term in $\Phi$ is convex, $\Phi$ is convex. Its left and right derivatives at any $s$ are
$\Phi'_-(s)=|\{e\in K:e<s\}|-\mathbf{1}_{\{s>0\}}-|B|$
and
$\Phi'_+(s)=|\{e\in K:e\leq s\}|-\mathbf{1}_{\{s\geq0\}}-|B|$,
respectively, where multiplicities are counted and an
indicator equals $1$ if its condition holds and $0$ otherwise. Since $s_0=e_{|B|+1}$, at most $|B|$ elements
are strictly smaller than $s_0$, while at least
$|B|+1$ elements are no greater than $s_0$. If $s_0=0$, then $\Phi'_-(s_0)=-|B|$;
otherwise, $\Phi'_-(s_0)\leq -1$. Hence, $\Phi'_-(s_0)\le0\le\Phi'_+(s_0)$. Therefore, $0$ belongs to the subdifferential of $\Phi$ at $s_0$, and $s_0$ minimizes $\Phi$ on $[0,1]$. Thus, $\Phi(s_0)\le\Phi(s)$ for all $s\in[0,1]$. Hence, Mechanism 1 minimizes the social cost.
\qed
\end{proof}

\subsection*{Proof of Lemma~\ref{lem:sc-sum-expansion}.}
\begin{proof}
Consider agents in $A$. For every $i\in M_2$, the claim follows directly from the corresponding analysis in the proof of Lemma~\ref{lem:sc-max-expansion}, since the cost function is identical.

If $i\in M_{12}$, then $C_i^{\Us}(s)=d_{i1}^{\Us}(s)+d_{i2}^{\Us}(s)$. If $x_i<0$, then $a_i=0$, so $C_i^{\Us}(s)=d_{i1}^{\Us}(s)-x_i+2s=(d_{i1}^{\Us}(s)-x_i)+2(s-a_i)^+$. If $0\le x_i<1$, then $a_i=x_i$. When $s\le x_i$, we have $(s-a_i)^+=0$, so $C_i^{\Us}(s)=d_{i1}^{\Us}(s)+x_i=(d_{i1}^{\Us}(s)+x_i)+2(s-a_i)^+$. When $s>x_i$, we have $C_i^{\Us}(s)=d_{i1}^{\Us}(s)+2s-x_i=(d_{i1}^{\Us}(s)+x_i)+2(s-a_i)^+$. If $x_i\ge1$, then $a_i=1$. Since $s\in[0,1]$, we have $(s-a_i)^+=0$, so $C_i^{\Us}(s)=d_{i1}^{\Us}(s)+x_i=(d_{i1}^{\Us}(s)+x_i)+2(s-a_i)^+$. Therefore, for every $i\in M_{12}$, there exists a constant $\alpha_i$ independent of $s$ such that $C_i^{\Us}(s)=\alpha_i+2(s-a_i)^+$.

The analysis for agents in $B$ is similar.
\qed
\end{proof}

\subsection*{Proof of Theorem~\ref{appstmt:3}.}
\begin{proof}
By Lemma~\ref{lem:sc-sum-expansion}, the individual cost functions under the sum-variant have the same functional form in terms of their respective breakpoints as those under the max-variant in Lemma~\ref{lem:sc-max-expansion}. Moreover, Mechanism~2 selects the $(|B|+1)$-st
element of $K$, as Mechanism~1 does. The strategyproofness proof of Theorem~1 relies only
on this form of the individual cost functions and
on the selection of the $(|B|+1)$-st element of $K$. Therefore, the same argument applies directly to Mechanism 2.

In particular, reporting interest only in the local facility removes the agent's breakpoint. Reporting a bridge-independent type while retaining interest in the opposite facility instead contributes a breakpoint at $1$ on $L_1$ or $0$ on $L_2$. These cases are covered by the proof of Theorem~1, with costs evaluated at the agent's true location and preference.

For optimality, Lemma~\ref{lem:sc-sum-expansion} implies that the social cost can be written as $\SC(s)=\SC_0+2\Phi(s)$, where $\SC_0$ is independent of $s$ and $\Phi(s)=\sum_{i\in A}(s-a_i)^++\sum_{j\in B}(b_j-s)^+$. This convex function has the same form as the one considered in the proof of Theorem~\ref{appstmt:1}. Since Mechanism~2 selects the $(|B|+1)$-st element
of $K$, the same left and right derivative argument
shows that its output minimizes $\Phi$, and hence
minimizes the social cost. Therefore, Mechanism 2 is strategyproof and optimal for the social cost objective under the sum-variant individual cost.
\qed
\end{proof}

\subsection*{Proof of Lemma~\ref{lem:sc-min-auxiliary}.}
\begin{proof}
Consider agents in $A_{\min}$. For every $i\in A_{\min}$, $d_{i2}^{\Us}(s)=|x_i-s|+s$. For every $i\in M_2$, the claim follows directly from the corresponding analysis in the proof of Lemma~\ref{lem:sc-max-expansion}.

Now consider $i\in\{h\in M_{12}:x_h<1/2\}$. If $x_i<0$, then $a_i=0$, so $d_{i2}^{\Us}(s)=-x_i+2s=(-x_i)+2(s-a_i)^+$. If $0\le x_i<1/2$, then $a_i=x_i$. When $s\le x_i$, we have $(s-a_i)^+=0$, so $d_{i2}^{\Us}(s)=x_i=x_i+2(s-a_i)^+$. When $s>x_i$, we have $d_{i2}^{\Us}(s)=2s-x_i=x_i+2(s-a_i)^+$. Therefore, for every $i\in\{h\in M_{12}:x_h<1/2\}$, there exists a constant $\alpha_i$ independent of $s$ such that $d_{i2}^{\Us}(s)=\alpha_i+2(s-a_i)^+$.

The analysis for agents in $B_{\min}$ is similar.
\qed
\end{proof}

\subsection*{Proof of Theorem~\ref{appstmt:5}.}
\begin{proof}
We first prove strategyproofness. By Remark~\ref{rem:outside-constant} and the definition of $A_{\min}$ and $B_{\min}$, all agents outside these sets have costs independent of the bridge position and cannot benefit from misreporting. Consider an agent $i\in A_{\min}$ and fix the reports of all other agents. Let $s_0$ be the output when she reports truthfully, and let $s'$ be the output after any misreport of her location and preference. Here, $A_{\min}$ and $B_{\min}$ are determined when agent $i$ reports truthfully. All cost comparisons use the deviating agent's true location and preference.

By Lemma~4, the crossing travel cost of agent $i$ is nondecreasing on $[0,1]$ and is minimized when $s\leq a_i$. Her true min-variant cost is either this crossing cost or its minimum with the constant local cost. Hence, her true cost has the same two properties. Since her misreport does not affect $B_{\min}$, both outputs are $0$ if $B_{\min}=\emptyset$. Now suppose $B_{\min}\neq\emptyset$. If $s_0\leq a_i$, the agent already attains her minimum cost. It remains to consider the case $a_i<s_0$.

If the agent remains in the reported set $A_{\min}$, her misreport only replaces $a_i$ with a reported breakpoint $a_i'\in[0,1]$. By the same order-statistic argument as in the proof of Theorem~1, the new output satisfies $s'\geq s_0$. If the misreport removes her from $A_{\min}$,
her breakpoint is deleted while the selected rank
$|B_{\min}|+1$ remains unchanged, so again $s'\geq s_0$. The deletion case includes reports with preference $(1,0)$ and reports with preference $(1,1)$ and $x_i'\geq 1/2$. Reports with preference $(0,1)$ and $x_i'\geq 1$ instead retain a breakpoint at $1$ and are covered by the replacement case. Since her true cost is nondecreasing, $C_i^U(s')\geq C_i^U(s_0)$ in both cases.

The argument for agents in $B_{\min}$ is analogous, using the fact that their true costs are nonincreasing and minimized when $s\geq b_j$. Reports with preference $(0,1)$ and reports with preference $(1,1)$ and $y_j'\leq 1/2$ remove the agent's breakpoint and reduce the selected rank by one. This cannot increase the output, with output $0$ if the reported set $B_{\min}$ becomes empty. Reports with preference $(1,0)$ and $y_j'\leq 0$ retain a breakpoint at $0$ and are covered by the replacement case. Therefore, no agent can benefit from misreporting, and Mechanism~3 is strategyproof.

We next analyze the approximation ratio. Using Lemma~\ref{lem:sc-min-auxiliary},  define
\[
\textstyle
H(s)=\sum_{i\in A_{\min}}\bigl(\alpha_i+2(s-a_i)^+\bigr)+\sum_{j\in B_{\min}}\bigl(\beta_j+2(b_j-s)^+\bigr),
\]
for every $s\in[0,1]$. Let $\SC_0=\sum_{i\in A_{\min}}\alpha_i+\sum_{j\in B_{\min}}\beta_j$, and let $\Phi(s)=\sum_{i\in A_{\min}}(s-a_i)^++\sum_{j\in B_{\min}}(b_j-s)^+$. Then $H(s)=\SC_0+2\Phi(s)$. Since $\SC_0$ does not depend on $s$, minimizing $H(s)$ is equivalent to minimizing $\Phi(s)$.

Mechanism~3 outputs $s_0=e_{|B_{\min}|+1}$. By the same argument as in the proof of Theorem~\ref{appstmt:1}, $s_0$ minimizes $\Phi(s)$ on $[0,1]$. Hence, $s_0$ also minimizes $H(s)$.

Let $s^*\in[0,1]$ be an optimal bridge position for the social cost. We next derive an upper bound on the social cost in terms of $H$. For every $s\in[0,1]$, we have
\[
\textstyle
\SC(s)=\sum_{i\in M\setminus A_{\min}} C_i^{\Us}(s)+\sum_{j\in N\setminus B_{\min}} C_j^{\Ls}(s)+\sum_{i\in A_{\min}} C_i^{\Us}(s)+\sum_{j\in B_{\min}} C_j^{\Ls}(s).
\]
Let $\SC_0'=\sum_{i\in M\setminus A_{\min}} C_i^{\Us}(s)+\sum_{j\in N\setminus B_{\min}} C_j^{\Ls}(s)$, which is independent of $s$ by the definitions of $A_{\min}$ and $B_{\min}$. Since each min-variant individual cost is no larger than the corresponding crossing travel cost, we have
\begin{equation}
\SC(s_0)=\SC_0'+\sum_{i\in A_{\min}}C_i^{\Us}(s_0)+\sum_{j\in B_{\min}}C_j^{\Ls}(s_0)\leq \SC_0'+H(s_0)\leq \SC_0'+H(s^*).
\label{eq:sc-min-upper}
\end{equation}
For $i\in M_2$, we have $C_i^{\Us}(s^*)=d_{i2}^{\Us}(s^*)$. For $i\in M_{12}\cap A_{\min}$, we have $x_i<1/2$. Since $s^*\in[0,1]$, $d_{i2}^{\Us}(s^*)=|x_i-s^*|+s^*\leq 2-x_i$. Since $x_i<1/2$, we have $2-x_i\leq3(1-x_i)=3d_{i1}^{\Us}(s^*)$. Together with $C_i^{\Us}(s^*)=\min\{d_{i1}^{\Us}(s^*),d_{i2}^{\Us}(s^*)\}$, this gives $d_{i2}^{\Us}(s^*)\leq3C_i^{\Us}(s^*)$ for every $i\in A_{\min}$. Therefore,
\begin{equation}
\sum_{i\in A_{\min}}d_{i2}^{\Us}(s^*)\leq3\sum_{i\in A_{\min}}C_i^{\Us}(s^*).
\label{eq:sc-min-A}
\end{equation}
The argument for $j\in B_{\min}$ is symmetric. Hence,
\begin{equation}
\sum_{j\in B_{\min}}d_{j1}^{\Ls}(s^*)\leq3\sum_{j\in B_{\min}}C_j^{\Ls}(s^*).
\label{eq:sc-min-B}
\end{equation}
Combining Eqs.~\eqref{eq:sc-min-upper}--\eqref{eq:sc-min-B} gives
\[
\begin{aligned}
\textstyle
\SC(s_0)&\leq \SC_0'+H(s^*)\\
&=\SC_0'+\sum_{i\in A_{\min}}d_{i2}^{\Us}(s^*)+\sum_{j\in B_{\min}}d_{j1}^{\Ls}(s^*)\\
&\leq \SC_0'+3\sum_{i\in A_{\min}}C_i^{\Us}(s^*)+3\sum_{j\in B_{\min}}C_j^{\Ls}(s^*)\\
&\leq 3\left(\SC_0'+\sum_{i\in A_{\min}}C_i^{\Us}(s^*)+\sum_{j\in B_{\min}}C_j^{\Ls}(s^*)\right)=3\SC(s^*).
\end{aligned}
\]

\noindent\textbf{Tight Example.}
Take any $\varepsilon\in(0,1/2)$. Consider an instance with two agents: an agent $i\in M_{12}$ on $L_1$ at $x_i=\frac12-\varepsilon$, and an agent $j\in N_1$ on $L_2$ at $y_j=1$. Under Mechanism 3, agent $i$ is counted in $A_{\min}$ with breakpoint $a_i=x_i$, while agent $j$ is counted in $B_{\min}$ with breakpoint $b_j=1$. Since $|B_{\min}|=1$, the mechanism outputs $s_0=a_i=\frac12-\varepsilon$. Then $C_i^{\Us}(s_0)=\min\{1-x_i,|x_i-s_0|+s_0\}=x_i=\frac12-\varepsilon$, and $C_j^{\Ls}(s_0)=|1-s_0|+1-s_0=2(1-s_0)=1+2\varepsilon$. Thus $\SC(s_0)=\frac32+\varepsilon$. On the other hand, at $s=1$, $C_i^{\Us}(1)=\min\{1-x_i,2-x_i\}=1-x_i=\frac12+\varepsilon$ and $C_j^{\Ls}(1)=0$, so $\OPT_{\SC}(\mathbf c)\le\frac12+\varepsilon$. Therefore $\SC(s_0)/\OPT_{\SC}(\mathbf c)\ge (3/2+\varepsilon)/(1/2+\varepsilon)\to 3$ as $\varepsilon\to0^+$.
\qed
\end{proof}

\subsection*{Proof of Theorem~\ref{thm:sc-min-both-lb2}.}
\begin{proof}
Assume that there exists a deterministic strategyproof mechanism $f$ whose approximation ratio is smaller than $2$. Consider three instances, each with two agents. In each instance, agent 1 is on $L_1$ at $x_1=0$, and agent 2 is on $L_2$ at $y_2=1$.

In instance $\mathbf c_1$, agent 1 only prefers $F_2$, and agent 2 prefers both facilities. Let $s_1=f(\mathbf c_1)$. For every bridge position $s$, the social cost is $\SC(\mathbf c_1,s)=2|s|+\min\{1,2|1-s|\}$. The optimal social cost is $1$, attained at $s=0$. Since $f$ has approximation ratio smaller than $2$, we have $\SC(\mathbf c_1,s_1)<2$. If $s_1\leq-1/2$, then $2|s_1|\geq1$ and $\min\{1,2|1-s_1|\}=1$, giving $\SC(\mathbf c_1,s_1)\geq2$. If $1/2\leq s_1\leq1$, then $\SC(\mathbf c_1,s_1)=2s_1+2(1-s_1)=2$. If $s_1>1$, then $2|s_1|>2$. All three cases contradict $\SC(\mathbf c_1,s_1)<2$. Hence, $-1/2<s_1<1/2$.

In instance $\mathbf c_2$, both agents prefer both facilities. Let $s_2=f(\mathbf c_2)$.

In instance $\mathbf c_3$, agent 1 prefers both facilities, and agent 2 only prefers $F_1$. Let $s_3=f(\mathbf c_3)$. For every bridge position $s$, the social cost is $\SC(\mathbf c_3,s)=\min\{1,2|s|\}+2|1-s|$. The optimal social cost is $1$, attained at $s=1$. Thus, $\SC(\mathbf c_3,s_3)<2$. If $s_3<0$, then $2|1-s_3|>2$. If $0\leq s_3\leq1/2$, then $\SC(\mathbf c_3,s_3)=2s_3+2(1-s_3)=2$. If $s_3\geq3/2$, then $\min\{1,2|s_3|\}=1$ and $2|1-s_3|\geq1$. Each case gives $\SC(\mathbf c_3,s_3)\geq2$, which is impossible. Therefore, $1/2<s_3<3/2$.

Now compare $\mathbf c_1$ and $\mathbf c_2$ for agent 1. In $\mathbf c_1$, if agent 1 reports that she prefers both facilities, the resulting reported instance is exactly $\mathbf c_2$, so the mechanism outputs $s_2$. Therefore, strategyproofness gives $2|s_1|\leq2|s_2|$, and hence $|s_1|\leq|s_2|$. Conversely, in $\mathbf c_2$, if she reports that she only prefers $F_2$, the resulting reported instance is exactly $\mathbf c_1$, so the mechanism outputs $s_1$. Therefore, strategyproofness gives $\min\{1,2|s_2|\}\leq\min\{1,2|s_1|\}$. Since $-1/2<s_1<1/2$, the right-hand side equals $2|s_1|<1$. Hence, $|s_2|<1/2$, and the inequality becomes $2|s_2|\leq2|s_1|$. Together with $|s_1|\leq|s_2|$, we obtain $|s_2|=|s_1|<1/2$.

Finally, consider agent 2 in $\mathbf c_2$. Her cost at bridge position $s$ is $\min\{1,2|1-s|\}$. Since $|s_2|<1/2$, we have $|1-s_2|>1/2$, so her cost at $s_2$ is $1$. If she reports that she only prefers $F_1$, the resulting reported instance is exactly $\mathbf c_3$, so the mechanism outputs $s_3$. Since $1/2<s_3<3/2$, we have $|1-s_3|<1/2$, and her cost after this report is $\min\{1,2|1-s_3|\}=2|1-s_3|<1$. Thus, agent 2 can reduce her cost by misreporting, contradicting strategyproofness.

Therefore, no deterministic strategyproof mechanism for the social cost objective can have an approximation ratio smaller than $2$. This completes the proof.
\qed
\end{proof}

\subsection*{Proof of Theorem~\ref{appstmt:9}.}
\begin{proof}
Assume that there exists a deterministic strategyproof mechanism $f$ whose approximation ratio is smaller than $5/3$.

Consider instance $\mathbf c_1$. Agent 1 is on $L_1$ at $x_1=-3/2$ and only prefers $F_2$; agent 2 is on $L_2$ at $y_2=5/2$ and only prefers $F_1$. Let $h=f(\mathbf c_1)$. In instance $\mathbf c_1$, for every bridge position $s$,  $C_1^{\mathrm U}(s)=|-3/2-s|+|s|$ and $C_2^{\mathrm L}(s)=|5/2-s|+|1-s|$. The optimal maximum cost of $\mathbf c_1$ is $5/2$, attained at $s=1/2$. Since the approximation ratio of $f$ is smaller than $5/3$, we have $\MC(\mathbf c_1,h)<25/6$. If $h\leq-1/3$, then $C_2^{\mathrm L}(h)=7/2-2h\geq25/6$, which contradicts $\MC(\mathbf c_1,h)<25/6$. If $h\geq4/3$, then $C_1^{\mathrm U}(h)=3/2+2h\geq25/6$, which also contradicts $\MC(\mathbf c_1,h)<25/6$. Therefore $-1/3<h<4/3$.

First suppose $-1/3<h\leq1/2$. Define instance $\mathbf c_2$ by changing only agent 1's location from $-3/2$ to $h$. Let $s_2=f(\mathbf c_2)$. In instance $\mathbf c_2$, $C_1^{\mathrm U}(s)=|h-s|+|s|$, whose minimum value is $|h|$, attained exactly on $[\min\{h,0\},\max\{h,0\}]$. The optimal maximum cost of $\mathbf c_2$ is $7/4-h/2$, attained at $s=7/8+h/4$. Hence strategyproofness implies $s_2\in[\min\{h,0\},\max\{h,0\}]$; otherwise agent 1 could report $-3/2$. Under this report, the mechanism receives the same reported profile as in $\mathbf c_1$ and therefore outputs $h$, reducing her cost to $|h|$.

If $-1/3<h\leq0$, then $s_2\leq0$, so $C_2^{\mathrm L}(s_2)=7/2-2s_2\geq7/2$. Thus $\MC(\mathbf c_2,s_2)/(7/4-h/2)\geq(7/2)/(7/4-h/2)>5/3$, where the last inequality follows from $h>-1/3$. If $0<h\leq1/2$, then $0\leq s_2\leq h$, so $C_2^{\mathrm L}(s_2)=7/2-2s_2\geq7/2-2h$. Hence $\MC(\mathbf c_2,s_2)/(7/4-h/2)\geq(7/2-2h)/(7/4-h/2)\geq5/3$, where the last inequality follows from $h\leq1/2$. In both cases, the approximation ratio of $f$ on $\mathbf c_2$ is at least $5/3$, contradicting the assumption that its approximation ratio is smaller than $5/3$.

The case $1/2<h<4/3$ is symmetric. By changing only agent 2's location from $5/2$ to $h$, the same argument gives an instance on which the approximation ratio of $f$ is at least $5/3$, contradicting the assumption that its approximation ratio is smaller than $5/3$.

Therefore, no deterministic strategyproof mechanism for the maximum cost objective can have an approximation ratio smaller than $5/3$. This completes the proof.
\qed
\end{proof}

\subsection*{Proof of Theorem~\ref{appstmt:7}.}
\begin{proof}
We first prove strategyproofness. By Remark~\ref{rem:outside-constant}, regardless of the mechanism's output, bridge-independent agents cannot benefit from misreporting. We show that regardless of the reports of all other agents, no agent $i\in A$ satisfying $x_i<1$ can benefit from misreporting. The analysis for any agent $j\in B$ satisfying $y_j>0$ is similar. For any such agent $i$, let $s_0$ be the output when she reports truthfully, and let $s'$ be the output after any misreport by $i$, while all other reports remain unchanged. Here, $A$ and $B$ are the sets determined when agent $i$ reports truthfully. All cost comparisons are based on the agent's true location and preference. Since agent $i$ is on $L_1$, her misreport does not affect $B$.

If $B=\varnothing$, the mechanism outputs $s_0=0$. Since $B$ remains unchanged after any misreport by agent $i$, the mechanism also outputs $s'=0$. Then agent $i$ attains her minimum cost on $[0,1]$ at $s_0=0$. Hence, the agent cannot benefit.

Now consider the case $B\ne\varnothing$. Since agent $i$ belongs to $A$ when she reports truthfully, the mechanism outputs $s_0=\operatorname{med}\{0,\tfrac{1}{2},1,a,b\}$. If $i\in M_2$, then $C_i^{\mathrm U}(s)=d_{i2}^{\mathrm U}(s)=|x_i-s|+s$; if $i\in M_{12}$, then $C_i^{\mathrm U}(s)=\max\{d_{i1}^{\mathrm U}(s),d_{i2}^{\mathrm U}(s)\}=\max\{1-x_i,|x_i-s|+s\}$. Hence, in either case, $C_i^{\mathrm U}(s)$ is minimized on $[0,1]$ whenever $s\le x_i$. Thus, if $s_0\le x_i$, the agent cannot benefit from misreporting. It remains to consider the case $s_0>x_i$.

If the agent's reported preference still contains $F_2$ after the misreport, then her report is still included when determining the new value $a'$, and only the value $a=\min_{h\in A}x_h$ may change. This includes reports with $x_i'\geq 1$, which correspond to bridge-independent types but retain the agent in the reported set $A$. Since her true location satisfies $x_i<1$, such a report cannot decrease the minimum reported location, so $a'\geq a$. If $a'\geq a$, the median form implies that the output cannot be smaller than $s_0$, so $s'\geq s_0$. If $a'<a$, then $a\leq x_i<s_0$, so $a$ is to the left of $s_0$; moving $a$ further left does not change the median output, hence $s'=s_0$. In either case, $s'\geq s_0$. Since $x_i<s_0\leq s'$, we have $C_i^{\mathrm U}(s')\geq C_i^{\mathrm U}(s_0)$. Hence, the agent cannot benefit.

If the agent's reported preference no longer contains $F_2$ after the misreport, then she is reported as belonging to $M_1$ and is no longer counted when determining the new value $a'$. If the reported set A remains nonempty, the new value $a'$ satisfies $a'\geq a$, and the median output cannot be smaller than $s_0$. If the reported set A becomes empty, the mechanism outputs $s'=1$, which is also not smaller than $s_0$. Hence, $s'\geq s_0$. Since $x_i<s_0\leq s'$, we have $C_i^{\mathrm U}(s')\geq C_i^{\mathrm U}(s_0)$. Thus, the agent cannot benefit.

We next analyze the approximation ratio. Given an instance $\mathbf c$, let $s_0$ be the output. If $B=\varnothing$, then $s_0=0$. At this position, every agent whose cost depends on the bridge position attains her minimum cost. Hence, the output is optimal. Similarly, if $A=\varnothing$ and $B\ne\varnothing$, then $s_0=1$, which is optimal. Thus assume $A\ne\varnothing$ and $B\ne\varnothing$.

Fix any $s\in[0,1]$. For every $i\in A$, $d_{i2}^{\Us}(s)=|x_i-s|+s$. If $s\le x_i$, then $d_{i2}^{\Us}(s)=x_i$; if $s>x_i$, then $d_{i2}^{\Us}(s)=2s-x_i\le2s-a$. Similarly, for every $j\in B$, $d_{j1}^{\Ls}(s)=|y_j-s|+1-s$. If $s\ge y_j$, then $d_{j1}^{\Ls}(s)=1-y_j$; if $s<y_j$, then $d_{j1}^{\Ls}(s)=y_j+1-2s\le b+1-2s$.

Combining these bounds with the local travel costs, and using the convention that the maximum over an empty set is $0$, we have
\[
\begin{aligned}
\MC(s)
&=\max\left\{
\begin{aligned}
&\max_{i\in M_1\cup M_{12}}|x_i-1|,
\max_{j\in N_2\cup N_{12}}|y_j|,\\
&\max_{i\in A}d_{i2}^{\Us}(s),
\max_{j\in B}d_{j1}^{\Ls}(s)
\end{aligned}
\right\}\\
&\le \max\left\{
\begin{aligned}
&\max_{i\in M_1\cup M_{12}}|x_i-1|,
\max_{j\in N_2\cup N_{12}}|y_j|,\\
&\max_{\substack{i\in A\\ x_i\ge0}}x_i,
\max_{\substack{j\in B\\ y_j\le1}}(1-y_j),
2s-a,
 b+1-2s
\end{aligned}
\right\}.
\end{aligned}
\]
Define
\[
\MC_0=\max\left\{
\begin{aligned}
&\max_{i\in M_1\cup M_{12}}|x_i-1|,
\max_{j\in N_2\cup N_{12}}|y_j|,\\
&\max_{\substack{i\in A\\ x_i\ge0}}x_i,
\max_{\substack{j\in B\\ y_j\le1}}(1-y_j)
\end{aligned}
\right\}.
\]
Thus, for every $s\in[0,1]$,
\begin{equation}
\MC(s)\le \max\{\MC_0,2s-a,b+1-2s\}.
\label{eq:max-upper}
\end{equation}
Hence, $\OPT_{\MC}(\mathbf c)\ge \MC_0$. If $\MC(s_0)\le\MC_0$, then $\OPT_{\MC}(\mathbf c)\le\MC(s_0)\le\MC_0\le\OPT_{\MC}(\mathbf c)$, so $\MC(s_0)=\OPT_{\MC}(\mathbf c)$. Hence, it remains to consider the case $\MC(s_0)>\MC_0$.

Let $s^*\in[0,1]$ be an optimal bridge position. Let $i_a\in A$ satisfy $x_{i_a}=a$, and let $j_b\in B$ satisfy $y_{j_b}=b$.
\begin{equation}
\OPT_{\MC}(\mathbf c)
\ge
\max\{d_{i_a,2}^{\Us}(s^*),d_{j_b,1}^{\Ls}(s^*)\}
\ge
\frac{d_{i_a,2}^{\Us}(s^*)+d_{j_b,1}^{\Ls}(s^*)}{2}
\ge
\frac{b+1-a}{2}.
\label{eq:max-pair-lb}
\end{equation}
If $a<0$, then $d_{i_a,2}^{\Us}(s)\ge -a$ for every $s\in[0,1]$. If $b>1$, then $d_{j_b,1}^{\Ls}(s)\ge b-1$ for every $s\in[0,1]$. Hence,
\begin{equation}
\OPT_{\MC}(\mathbf c)\ge\max\{-a,b-1,0\}.
\label{eq:max-outside-lb}
\end{equation}

We distinguish the following cases according to the relative positions of $a$, $b$, $0$, $1/2$, and $1$.

\textbf{Case 1.} $a\le1/2\le b$. Then $s_0=1/2$. Since $\MC(s_0)>\MC_0$, Eq.~\eqref{eq:max-upper} gives $\MC(s_0)\le\max\{1-a,b\}$.

If $1-a\ge b$, then $\MC(s_0)\le1-a$. When $a\le-3/2$, Eq.~\eqref{eq:max-outside-lb} gives $\OPT_{\MC}(\mathbf c)\ge-a$. Hence,
\[
\frac{\MC(s_0)}{\OPT_{\MC}(\mathbf c)}
\le
\frac{1-a}{-a}
=
1+\frac{1}{-a}
\le
\frac53.
\]
When $a>-3/2$, Eq.~\eqref{eq:max-pair-lb} and $b\ge1/2$ give $\OPT_{\MC}(\mathbf c)\ge(b+1-a)/2\ge(3/2-a)/2$. Hence,
\[
\frac{\MC(s_0)}{\OPT_{\MC}(\mathbf c)}
\le
\frac{1-a}{(3/2-a)/2}
=
\frac{2(1-a)}{3/2-a}
\le
\frac53.
\]

If $1-a<b$, then $\MC(s_0)\le b$. When $b\le5/2$, Eq.~\eqref{eq:max-pair-lb} and $a\le1/2$ give $\OPT_{\MC}(\mathbf c)\ge(b+1-a)/2\ge(b+1/2)/2$. Hence,
\[
\frac{\MC(s_0)}{\OPT_{\MC}(\mathbf c)}
\le
\frac{b}{(b+1/2)/2}
=
\frac{2b}{b+1/2}
\le
\frac53.
\]
When $b>5/2$, Eq.~\eqref{eq:max-outside-lb} gives $\OPT_{\MC}(\mathbf c)\ge b-1$. Hence,
\[
\frac{\MC(s_0)}{\OPT_{\MC}(\mathbf c)}
\le
\frac{b}{b-1}
=
1+\frac{1}{b-1}
<
\frac53.
\]
Therefore, when $a\le1/2\le b$, we have $\MC(s_0)/\OPT_{\MC}(\mathbf c)\le5/3$.

\textbf{Case 2.} $a\le b\le0$. Then $s_0=0$. At this position, every agent whose cost depends on the bridge position attains her minimum cost. Hence, the output is optimal.

\textbf{Case 3.} $a\le b$ and $0<b\le1/2$. Then $s_0=b$. Every agent in $B$ attains her minimum cost, and every agent $i\in A$ with $x_i\ge b$ also attains her minimum cost. Since $\MC(s_0)>\MC_0$, the maximum cost at $s_0$ can only be attained by an agent $i\in A$ with $x_i<b$. Eq.~\eqref{eq:max-upper} gives $\MC(s_0)\le2b-a$.

If $a\le-3/2$, then Eq.~\eqref{eq:max-outside-lb} gives $\OPT_{\MC}(\mathbf c)\ge-a$. Hence,
\[
\frac{\MC(s_0)}{\OPT_{\MC}(\mathbf c)}
\le
\frac{2b-a}{-a}
=
1+\frac{2b}{-a}
\le
1+\frac{1}{-a}
\le
\frac53.
\]
If $a>-3/2$, then Eq.~\eqref{eq:max-pair-lb} gives $\OPT_{\MC}(\mathbf c)\ge(b+1-a)/2$. Hence,
\[
\frac{\MC(s_0)}{\OPT_{\MC}(\mathbf c)}
\le
\frac{2b-a}{(b+1-a)/2}
=
\frac{2(2b-a)}{b+1-a}
\le
\frac53,
\]
where the last inequality follows from $b\le1/2$ and $a>-3/2$.

\textbf{Case 4.} $1/2\le a\le b$. Then $s_0=\min\{1,a\}$. If $a\ge1$, then $s_0=1$, and every agent whose cost depends on the bridge position attains her minimum cost. Hence, the output is optimal. Otherwise, $s_0=a$, and this case is symmetric to Case 3. Hence, $\MC(s_0)/\OPT_{\MC}(\mathbf c)\le5/3$.

\textbf{Case 5.} $b\le1/2\le a$. Then $s_0=1/2$, and all agents in $A$ and $B$ whose costs depend on the bridge position attain their minimum costs. Hence, the output is optimal.

\textbf{Case 6.} $b\le a\le1/2$. Then $s_0=\max\{0,a\}$, and all agents in $A$ and $B$ whose costs depend on the bridge position attain their minimum costs. Hence, the output is optimal.

\textbf{Case 7.} $1/2\le b\le a$. Then $s_0=\min\{1,b\}$, and this case is symmetric to Case 6. Hence, the output is optimal.

Combining all cases, we have $\MC(s_0)\le(5/3)\OPT_{\MC}(\mathbf c)$.

\noindent\textbf{Tight Example.}
For any $\varepsilon\in(0,1)$, consider an instance with two agents: one agent $i\in M_2$ located at $x_i=-\frac32+\varepsilon$, and one agent $j\in N_1$ located at $y_j=\frac12$. Then $a=-\frac32+\varepsilon$ and $b=\frac12$, so $s_0=\frac12$. At this output, $C_i^{\mathrm U}(s_0)=\left|-\frac32+\varepsilon-\frac12\right|+\frac12=\frac52-\varepsilon$, while $C_j^{\mathrm L}(s_0)=\frac12$. Hence, $\MC(s_0)=\frac52-\varepsilon$. On the other hand, at $s=\varepsilon/4$, both crossing travel costs are $\frac32-\varepsilon/2$, so $\OPT_{\MC}(\mathbf c)\le\frac32-\varepsilon/2$. Therefore, $\MC(s_0)/\OPT_{\MC}(\mathbf c)\ge(\frac52-\varepsilon)/(\frac32-\varepsilon/2)$, which tends to $5/3$ as $\varepsilon\to0^+$.
\qed
\end{proof}

\subsection*{Proof of Theorem~\ref{appstmt:10}.}
\begin{proof}
We first prove strategyproofness. By Remark~\ref{rem:outside-constant}, bridge-independent agents cannot benefit from misreporting. Under the sum-variant individual cost, for every agent in $A$ and $B$, her cost can be written as the corresponding crossing travel cost plus a term independent of the bridge position. Therefore, her cost has the same monotonicity with respect to the bridge position as in the proof of Theorem~\ref{appstmt:7}. Since Mechanism~4 uses the same median rule, the strategyproofness argument in the proof of Theorem~6 applies, including reports that remove an agent from $A$ or $B$ and reports of bridge-independent types that retain the agent in the corresponding set. All cost comparisons are based on the agent's true location and preference.

We next analyze the approximation ratio. Given an instance $\mathbf c$, let $s_0$ be the output. If $B=\varnothing$, then $s_0=0$. At this position, every agent whose cost depends on the bridge position attains her minimum cost. Hence, the output is optimal. Similarly, if $A=\varnothing$ and $B\ne\varnothing$, then $s_0=1$, which is optimal. Thus assume $A\ne\varnothing$ and $B\ne\varnothing$.

For each nonempty set among $M_2$, $M_{12}$, $N_1$, and $N_{12}$, define the corresponding extreme location by $x_\ell^2=\min_{i\in M_2}x_i$, $x_\ell^{12}=\min_{i\in M_{12}}x_i$, $y_r^1=\max_{j\in N_1}y_j$, and $y_r^{12}=\max_{j\in N_{12}}y_j$, respectively.

Fix any $s\in[0,1]$. For $i\in M_2$, if $s\le x_i$, then $C_i^{\Us}(s)=d_{i2}^{\Us}(s)=x_i$, while if $s>x_i$, then $C_i^{\Us}(s)=d_{i2}^{\Us}(s)=2s-x_i\le2s-x_\ell^2$. For $i\in M_{12}$, if $x_i<1$, then $C_i^{\Us}(s)=d_{i1}^{\Us}(s)+d_{i2}^{\Us}(s)=\max\{1,1+2s-2x_i\}$, while if $x_i\ge1$, then $C_i^{\Us}(s)=2x_i-1$, which is independent of $s$. Similarly, for $j\in N_1$, if $s\ge y_j$, then $C_j^{\Ls}(s)=d_{j1}^{\Ls}(s)=1-y_j$, while if $s<y_j$, then $C_j^{\Ls}(s)=d_{j1}^{\Ls}(s)=y_j+1-2s\le y_r^1+1-2s$. For $j\in N_{12}$, if $y_j>0$, then $C_j^{\Ls}(s)=d_{j2}^{\Ls}(s)+d_{j1}^{\Ls}(s)=\max\{1,2y_j+1-2s\}$, while if $y_j\le0$, then $C_j^{\Ls}(s)=1-2y_j$, which is independent of $s$. With the convention that the maximum over an empty set is $0$, we have
\[
\begin{aligned}
\MC(s)
&=\max\left\{
\begin{aligned}
&\max_{i\in M_1}|x_i-1|,
\max_{j\in N_2}|y_j|,
\max_{i\in M_2}d_{i2}^{\Us}(s),\\
&\max_{j\in N_1}d_{j1}^{\Ls}(s),
\max_{i\in M_{12}}\bigl(d_{i1}^{\Us}(s)+d_{i2}^{\Us}(s)\bigr),\\
&\max_{j\in N_{12}}\bigl(d_{j2}^{\Ls}(s)+d_{j1}^{\Ls}(s)\bigr)
\end{aligned}
\right\}\\
&\le \max\left\{
\begin{aligned}
&\max_{i\in M_1}|x_i-1|,
\max_{j\in N_2}|y_j|,
\max_{\substack{i\in M_2\\ x_i\ge0}}x_i,\\
&\max_{\substack{j\in N_1\\ y_j\le1}}(1-y_j),
\max_{k\in M_{12}\cup N_{12}}1,
\max_{\substack{i\in M_{12}\\ x_i\ge1}}(2x_i-1),\\
&\max_{\substack{j\in N_{12}\\ y_j\le0}}(1-2y_j),
2s-x_\ell^2,
1+2s-2x_\ell^{12},\\
&y_r^1+1-2s,
2y_r^{12}+1-2s
\end{aligned}
\right\}.
\end{aligned}
\]
Define
\[
\begin{aligned}
\MC_0=\max\{&
\max_{i\in M_1}|x_i-1|,
\max_{j\in N_2}|y_j|,
\max_{\substack{i\in M_2\\ x_i\ge0}}x_i,
\max_{\substack{j\in N_1\\ y_j\le1}}(1-y_j),\\
&
\max_{k\in M_{12}\cup N_{12}}1,
\max_{\substack{i\in M_{12}\\ x_i\ge1}}(2x_i-1),
\max_{\substack{j\in N_{12}\\ y_j\le0}}(1-2y_j)\}.
\end{aligned}
\]
Thus, for every $s\in[0,1]$,
\begin{equation}
\begin{aligned}
\MC(s)\le \max\{&
\MC_0,
2s-x_\ell^2,
1+2s-2x_\ell^{12},\\
&
y_r^1+1-2s,
2y_r^{12}+1-2s\}.
\end{aligned}
\label{eq:sum-upper}
\end{equation}
Hence, $\OPT_{\MC}(\mathbf c)\ge \MC_0$. If $\MC(s_0)\le\MC_0$, then $\OPT_{\MC}(\mathbf c)\le\MC(s_0)\le\MC_0\le\OPT_{\MC}(\mathbf c)$, so $\MC(s_0)=\OPT_{\MC}(\mathbf c)$. Hence, it remains to consider the case $\MC(s_0)>\MC_0$.

Let $s^*\in[0,1]$ be an optimal bridge position. Whenever the relevant set is nonempty, choose an agent attaining the corresponding extreme location: let $i_2\in M_2$ satisfy $x_{i_2}=x_\ell^2$, let $i_{12}\in M_{12}$ satisfy $x_{i_{12}}=x_\ell^{12}$, let $j_1\in N_1$ satisfy $y_{j_1}=y_r^1$, and let $j_{12}\in N_{12}$ satisfy $y_{j_{12}}=y_r^{12}$.

If $M_2$ and $N_1$ are nonempty, then, by considering agents $i_2$ and $j_1$, we have
\begin{equation}
\OPT_{\MC}(\mathbf c)
\ge
\frac{|x_\ell^2-s^*|+s^*+|y_r^1-s^*|+1-s^*}{2}
\ge
\frac{1+y_r^1-x_\ell^2}{2}.
\label{eq:sum-M2-N1}
\end{equation}
If $M_2$ and $N_{12}$ are nonempty, then, by considering agents $i_2$ and $j_{12}$, we have
\begin{equation}
\OPT_{\MC}(\mathbf c)
\ge
\frac{|x_\ell^2-s^*|+s^*+y_r^{12}+|y_r^{12}-s^*|+1-s^*}{2}
\ge
\frac{1+2y_r^{12}-x_\ell^2}{2}.
\label{eq:sum-M2-N12}
\end{equation}
If $M_{12}$ and $N_1$ are nonempty, then, by considering agents $i_{12}$ and $j_1$, we have
\begin{equation}
\OPT_{\MC}(\mathbf c)
\ge
\frac{1-x_\ell^{12}+|x_\ell^{12}-s^*|+s^*+|y_r^1-s^*|+1-s^*}{2}
\ge
1+\frac{y_r^1}{2}-x_\ell^{12}.
\label{eq:sum-M12-N1}
\end{equation}
If $M_{12}$ and $N_{12}$ are nonempty, then, by considering agents $i_{12}$ and $j_{12}$, we have
\begin{equation}
\begin{aligned}
\OPT_{\MC}(\mathbf c)
&\ge
\frac{1-x_\ell^{12}+|x_\ell^{12}-s^*|+s^*+y_r^{12}+|y_r^{12}-s^*|+1-s^*}{2}\\
&\ge
1+y_r^{12}-x_\ell^{12}.
\end{aligned}
\label{eq:sum-M12-N12}
\end{equation}
Moreover, whenever the corresponding sets are nonempty, the boundary agents give $\OPT_{\MC}(\mathbf c)\ge -x_\ell^2$, $\OPT_{\MC}(\mathbf c)\ge1-2x_\ell^{12}$, $\OPT_{\MC}(\mathbf c)\ge y_r^1-1$, and $\OPT_{\MC}(\mathbf c)\ge2y_r^{12}-1$. We refer to these four inequalities collectively as the boundary lower bounds.

We distinguish the following cases according to the relative positions of $a$, $b$, $0$, $1/2$, and $1$.

\textbf{Case 1.} $a\le1/2\le b$. Then $s_0=1/2$. Since $\MC(s_0)>\MC_0$, Eq.~\eqref{eq:sum-upper} gives $\MC(s_0)\le\max\{1-x_\ell^2,2-2x_\ell^{12},y_r^1,2y_r^{12}\}$, where terms corresponding to empty sets are omitted.

\textbf{Case 1.1.} Consider agents in $M_2$. If the maximum cost is attained by an agent $i\in M_2$, then $x_i<1/2$, since otherwise $i$ attains her minimum cost. In this case, $C_i^{\Us}(1/2)=1-x_i\le1-x_\ell^2$. The boundary lower bound gives $\OPT_{\MC}(\mathbf c)\ge -x_\ell^2$. If $b=y_r^1$, then Eq.~\eqref{eq:sum-M2-N1} gives $\OPT_{\MC}(\mathbf c)\ge(1+b-x_\ell^2)/2$. If $b=y_r^{12}$, then Eq.~\eqref{eq:sum-M2-N12} gives $\OPT_{\MC}(\mathbf c)\ge(1+2b-x_\ell^2)/2\ge(1+b-x_\ell^2)/2$. Hence, in either case, $\OPT_{\MC}(\mathbf c)\ge(1+b-x_\ell^2)/2$.

If $x_\ell^2\le -(1+b)$, then
\[
\frac{\MC(s_0)}{\OPT_{\MC}(\mathbf c)}
\le
\frac{1-x_\ell^2}{-x_\ell^2}
=
1+\frac{1}{-x_\ell^2}
\le
1+\frac{1}{1+b}
\le
\frac53.
\]
If $x_\ell^2>-(1+b)$, then
\[
\frac{\MC(s_0)}{\OPT_{\MC}(\mathbf c)}
\le
\frac{2(1-x_\ell^2)}{1+b-x_\ell^2}
\le
\frac53.
\]
Thus, this case satisfies $\MC(s_0)\le(5/3)\OPT_{\MC}(\mathbf c)$.

\textbf{Case 1.2.} Consider agents in $M_{12}$. If the maximum cost is attained by an agent $i\in M_{12}$, then $x_i<1/2$, since otherwise $i$ attains her minimum cost. In this case, $C_i^{\Us}(1/2)=2-2x_i\le2-2x_\ell^{12}$. The boundary lower bound gives $\OPT_{\MC}(\mathbf c)\ge1-2x_\ell^{12}$. If $b=y_r^1$, then Eq.~\eqref{eq:sum-M12-N1} gives $\OPT_{\MC}(\mathbf c)\ge1+b/2-x_\ell^{12}$. If $b=y_r^{12}$, then Eq.~\eqref{eq:sum-M12-N12} gives $\OPT_{\MC}(\mathbf c)\ge1+b-x_\ell^{12}\ge1+b/2-x_\ell^{12}$. Hence, in either case, $\OPT_{\MC}(\mathbf c)\ge1+b/2-x_\ell^{12}$.

If $x_\ell^{12}\le -b/2$, then
\[
\frac{\MC(s_0)}{\OPT_{\MC}(\mathbf c)}
\le
\frac{2-2x_\ell^{12}}{1-2x_\ell^{12}}
=
1+\frac{1}{1-2x_\ell^{12}}
\le
1+\frac{1}{1+b}
\le
\frac53.
\]
If $-b/2<x_\ell^{12}<1/2$, then
\[
\frac{\MC(s_0)}{\OPT_{\MC}(\mathbf c)}
\le
\frac{2-2x_\ell^{12}}{1+b/2-x_\ell^{12}}
\le
\frac53.
\]
Thus, this case satisfies $\MC(s_0)\le(5/3)\OPT_{\MC}(\mathbf c)$.

\textbf{Case 1.3.} Consider agents in $N_1$. If the maximum cost is attained by an agent $j\in N_1$, then $y_j>1/2$, since otherwise $j$ attains her minimum cost. In this case, $C_j^{\Ls}(1/2)=y_j\le y_r^1$. The boundary lower bound gives $\OPT_{\MC}(\mathbf c)\ge y_r^1-1$. If $a=x_\ell^2$, then Eq.~\eqref{eq:sum-M2-N1} gives $\OPT_{\MC}(\mathbf c)\ge(1+y_r^1-a)/2\ge(y_r^1+1/2)/2$. If $a=x_\ell^{12}$, then Eq.~\eqref{eq:sum-M12-N1} gives $\OPT_{\MC}(\mathbf c)\ge1+y_r^1/2-a\ge(y_r^1+1/2)/2$. Hence, in either case, $\OPT_{\MC}(\mathbf c)\ge(y_r^1+1/2)/2$.

If $1/2\le y_r^1\le5/2$, then
\[
\frac{\MC(s_0)}{\OPT_{\MC}(\mathbf c)}
\le
\frac{2y_r^1}{y_r^1+1/2}
\le
\frac53.
\]
If $y_r^1>5/2$, then
\[
\frac{\MC(s_0)}{\OPT_{\MC}(\mathbf c)}
\le
\frac{y_r^1}{y_r^1-1}
\le
\frac53.
\]
Thus, this case satisfies $\MC(s_0)\le(5/3)\OPT_{\MC}(\mathbf c)$.

\textbf{Case 1.4.} Consider agents in $N_{12}$. If the maximum cost is attained by an agent $j\in N_{12}$, then $y_j>1/2$, since otherwise $j$ attains her minimum cost. In this case, $C_j^{\Ls}(1/2)=2y_j\le2y_r^{12}$. The boundary lower bound gives $\OPT_{\MC}(\mathbf c)\ge2y_r^{12}-1$. If $a=x_\ell^2$, then Eq.~\eqref{eq:sum-M2-N12} gives $\OPT_{\MC}(\mathbf c)\ge(1+2y_r^{12}-a)/2\ge y_r^{12}+1/4$. If $a=x_\ell^{12}$, then Eq.~\eqref{eq:sum-M12-N12} gives $\OPT_{\MC}(\mathbf c)\ge1+y_r^{12}-a\ge y_r^{12}+1/4$. Hence, in either case, $\OPT_{\MC}(\mathbf c)\ge y_r^{12}+1/4$.

If $1/2\le y_r^{12}\le5/4$, then
\[
\frac{\MC(s_0)}{\OPT_{\MC}(\mathbf c)}
\le
\frac{2y_r^{12}}{y_r^{12}+1/4}
\le
\frac53.
\]
If $y_r^{12}>5/4$, then
\[
\frac{\MC(s_0)}{\OPT_{\MC}(\mathbf c)}
\le
\frac{2y_r^{12}}{2y_r^{12}-1}
\le
\frac53.
\]
Thus, this case satisfies $\MC(s_0)\le(5/3)\OPT_{\MC}(\mathbf c)$. Therefore, when $a\le1/2\le b$, we have $\MC(s_0)\le(5/3)\OPT_{\MC}(\mathbf c)$.

\textbf{Case 2.} $a\le b\le0$. Then $s_0=0$. At this position, every agent whose cost depends on the bridge position attains her minimum cost. Hence, the output is optimal.

\textbf{Case 3.} $a\le b$ and $0<b\le1/2$. Then $s_0=b$. Every agent in $B$ attains her minimum cost, and every agent $i\in A$ with $x_i\ge b$ also attains her minimum cost. Moreover, the two terms in Eq.~\eqref{eq:sum-upper} corresponding to agents in $B$ are at most $\MC_0$. Hence, since $\MC(s_0)>\MC_0$, Eq.~\eqref{eq:sum-upper} gives $\MC(s_0)\le\max\{2b-x_\ell^2,1+2b-2x_\ell^{12}\}$. Therefore, it remains to consider agents $i\in M_2$ with $x_i<b$ and agents $i\in M_{12}$ with $x_i<b$.

\textbf{Case 3.1.} Suppose the maximum cost is attained by an agent $i\in M_2$ with $x_i<b$. For such an agent, $C_i^{\Us}(b)=2b-x_i\le2b-x_\ell^2$. The boundary lower bound gives $\OPT_{\MC}(\mathbf c)\ge -x_\ell^2$. If $b=y_r^1$, then Eq.~\eqref{eq:sum-M2-N1} gives $\OPT_{\MC}(\mathbf c)\ge(1+b-x_\ell^2)/2$. If $b=y_r^{12}$, then Eq.~\eqref{eq:sum-M2-N12} gives $\OPT_{\MC}(\mathbf c)\ge(1+2b-x_\ell^2)/2\ge(1+b-x_\ell^2)/2$. Hence, in either case, $\OPT_{\MC}(\mathbf c)\ge(1+b-x_\ell^2)/2$.

If $x_\ell^2\le -(1+b)$, then
\[
\frac{\MC(s_0)}{\OPT_{\MC}(\mathbf c)}
\le
\frac{2b-x_\ell^2}{-x_\ell^2}
=
1+\frac{2b}{-x_\ell^2}
\le
1+\frac{2b}{1+b}
=
\frac{1+3b}{1+b}
\le
\frac53.
\]
If $x_\ell^2>-(1+b)$, then
\[
\frac{\MC(s_0)}{\OPT_{\MC}(\mathbf c)}
\le
\frac{2(2b-x_\ell^2)}{1+b-x_\ell^2}
\le
\frac53.
\]
Thus, this case satisfies $\MC(s_0)\le(5/3)\OPT_{\MC}(\mathbf c)$.

\textbf{Case 3.2.} Suppose the maximum cost is attained by an agent $i\in M_{12}$ with $x_i<b$. For such an agent, $C_i^{\Us}(b)=1+2b-2x_i\le1+2b-2x_\ell^{12}$. The boundary lower bound gives $\OPT_{\MC}(\mathbf c)\ge1-2x_\ell^{12}$. If $b=y_r^1$, then Eq.~\eqref{eq:sum-M12-N1} gives $\OPT_{\MC}(\mathbf c)\ge1+b/2-x_\ell^{12}$; if $b=y_r^{12}$, then Eq.~\eqref{eq:sum-M12-N12} gives $\OPT_{\MC}(\mathbf c)\ge1+b-x_\ell^{12}$. Hence, in both cases, $\OPT_{\MC}(\mathbf c)\ge1+b/2-x_\ell^{12}$.

If $x_\ell^{12}\le -b/2$, then
\[
\frac{\MC(s_0)}{\OPT_{\MC}(\mathbf c)}
\le
\frac{1+2b-2x_\ell^{12}}{1-2x_\ell^{12}}
=
1+\frac{2b}{1-2x_\ell^{12}}
\le
1+\frac{2b}{1+b}
=
\frac{1+3b}{1+b}
\le
\frac53.
\]
If $-b/2<x_\ell^{12}<b$, then the upper bound on $\MC(s_0)$ and the lower bound on $\OPT_{\MC}(\mathbf c)$ give
\[
\frac{\MC(s_0)}{\OPT_{\MC}(\mathbf c)}
\le
\frac{1+2b-2x_\ell^{12}}{1+b/2-x_\ell^{12}}.
\]
To bound the right-hand side, define $\Phi(x)=(1+2b-2x)/(1+b/2-x)$. Since $\Phi'(x)=(b-1)/(1+b/2-x)^2<0$ for $0<b\le1/2$, $\Phi(x)$ is decreasing on $(-b/2,b)$. Hence,
\[
\frac{\MC(s_0)}{\OPT_{\MC}(\mathbf c)}
\le
\Phi(x_\ell^{12})
<
\Phi(-b/2)
=
\frac{1+3b}{1+b}
\le
\frac53.
\]
Therefore, when $a\le b$ and $0<b\le1/2$, we have $\MC(s_0)\le(5/3)\OPT_{\MC}(\mathbf c)$.

\textbf{Case 4.} $1/2\le a\le b$. Then $s_0=\min\{1,a\}$. If $a\ge1$, then $s_0=1$, and every agent whose cost depends on the bridge position attains her minimum cost. Hence, the output is optimal. Otherwise, $s_0=a$, and this case is symmetric to Case 3. Hence, $\MC(s_0)\le(5/3)\OPT_{\MC}(\mathbf c)$.

\textbf{Case 5.} $b\le1/2\le a$. Then $s_0=1/2$. In this case, all agents in $A$ and $B$ whose costs depend on the bridge position attain their minimum costs. Together with the fixed term $\MC_0$, the output is optimal.

\textbf{Case 6.} $b\le a\le1/2$. Then $s_0=\max\{0,a\}$. If $a\le0$, then $s_0=0$; if $0<a\le1/2$, then $s_0=a$. In both cases, all agents in $A$ and $B$ whose costs depend on the bridge position attain their minimum costs. Hence, the output is optimal.

\textbf{Case 7.} $1/2\le b\le a$. Then $s_0=\min\{1,b\}$. This case is symmetric to Case 6. Hence, the output is optimal.

Combining all cases, we have $\MC(s_0)\le(5/3)\OPT_{\MC}(\mathbf c)$.

The tight example following Theorem 6 also applies here, since every agent in
that instance is interested in exactly one facility.

\qed
\end{proof}

\subsection*{Proof of Lemma~\ref{lem:minmax-control-lb}.}
\begin{proof}
We prove the first statement. If $s_0=0$, the claim is clear. Assume $s_0>0$. Since $x_i<s_0$, we have $A\neq\emptyset$. If $B=\emptyset$, then the mechanism outputs $0$, a contradiction. Hence $A\neq\emptyset$ and $B\neq\emptyset$, and $s_0=\operatorname{med}\{0,\tfrac{1}{2},1,a,b\}$. Since $0<s_0$ and $a\leq x_i<s_0$, the entries $0$ and $a$ are strictly smaller than $s_0$. As $s_0$ is the third smallest of the five entries, each of the remaining entries $1/2$, $1$, and $b$ is at least $s_0$. Hence, $s_0\leq 1/2$ and $b\geq s_0$. Let $j_b\in B$ satisfy $y_{j_b}=b$.

Take any bridge position $s\in[0,1]$. If $s>s_0$, then $x_i<s_0<s$. If $i\in M_2$, then $C_i^{\mathrm U}(s)=2s-x_i>s_0$. If $i\in M_{12}$, then $x_i<s_0\le1/2$, so $1-x_i\ge s_0$, and $|x_i-s|+s=2s-x_i>s_0$. Hence $C_i^{\mathrm U}(s)\ge s_0$, and thus $\MC(s)\ge s_0$.

If $s\le s_0$, then $s\le s_0\le b$. If $j_b\in N_1$, then $C_{j_b}^{\mathrm L}(s)=b+1-2s\ge b\ge s_0$. If $j_b\in N_{12}$, then $s_0\le1/2$ and $b\ge s_0$, so her min-variant cost is also at least $s_0$. Hence $\MC(s)\ge s_0$. Therefore $\MC(s)\ge s_0$ for every $s\in[0,1]$, and $\OPT_{\MC}(\mathbf c)\ge s_0$.

The second statement follows by a symmetric argument.
\qed
\end{proof}

\subsection*{Proof of Theorem~\ref{appstmt:12}.}
\begin{proof}
We first prove strategyproofness. By Remark~\ref{rem:outside-constant}, bridge-independent agents cannot benefit from misreporting. In addition, agents $i\in M_{12}$ with $1/2\le x_i<1$ and agents $j\in N_{12}$ with $0<y_j\le1/2$ have costs independent of the bridge position and hence cannot benefit from misreporting. For each remaining agent on $L_1$, her true cost is nondecreasing on $[0,1]$ and is minimized whenever $s\leq x_i$; for each remaining agent on $L_2$, it is nonincreasing and is minimized whenever $s\geq y_j$. Since Mechanism~4 uses the same median rule, the strategyproofness argument in the proof of Theorem~6 applies, including reports that remove an agent from $A$ or $B$ and reports that retain the agent in the corresponding set. In particular, reporting preference $(1,1)$ with $x_i'\geq 1/2$ on $L_1$ or $y_j'\leq 1/2$ on $L_2$ specifies a bridge-independent type but does not remove the agent from $A$ or $B$, respectively. Such reports are covered by the retained-report case, with costs evaluated at the agent's true location and preference.

We next analyze the approximation ratio.

Given an instance $\mathbf c$, let $s_0$ be the output of Mechanism 4, and let $s^*$ be an optimal bridge position. By Remark~\ref{rem:outside-constant}, bridge-independent agents have the same cost at $s_0$ and at $s^*$. Hence, if the maximum cost at $s_0$ is attained by a bridge-independent agent, then $\MC(s_0)\le\OPT_{\MC}(\mathbf c)$, and the ratio is at most $1$. Therefore, we only need to consider agents whose costs may depend on the bridge position. If one of $A$ and $B$ is empty, the endpoint output is optimal. Thus assume $A\ne\varnothing$ and $B\ne\varnothing$.

We first give some lower bounds on $\OPT_{\MC}(\mathbf c)$ for individual agents. For $i\in M_2$, $C_i^{\mathrm U}(s^*)=|x_i-s^*|+s^*\ge |x_i|$, and hence $\OPT_{\MC}(\mathbf c)\ge |x_i|$. For $i\in M_{12}$, $C_i^{\mathrm U}(s^*)=\min\{|x_i-1|,|x_i-s^*|+s^*\}\ge \min\{|x_i-1|,|x_i|\}$, and hence $\OPT_{\MC}(\mathbf c)\ge\min\{|x_i-1|,|x_i|\}$. Similarly, for $j\in N_1$, $\OPT_{\MC}(\mathbf c)\ge |1-y_j|$, and for $j\in N_{12}$, $\OPT_{\MC}(\mathbf c)\ge\min\{|y_j|,|1-y_j|\}$.

Consider agents on $L_1$. If $i\in M_2$, then $C_i^{\mathrm U}(s_0)=|x_i-s_0|+s_0$. If $x_i\ge s_0$, then $C_i^{\mathrm U}(s_0)=x_i=|x_i|\le\OPT_{\MC}(\mathbf c)$. If $x_i<s_0$, then by Lemma~\ref{lem:minmax-control-lb}, $\OPT_{\MC}(\mathbf c)\ge s_0$. Also, by the lower bound above, $\OPT_{\MC}(\mathbf c)\ge |x_i|$. Thus
\[
C_i^{\mathrm U}(s_0)=2s_0-x_i\le2s_0+|x_i|\le3\OPT_{\MC}(\mathbf c).
\]

If $i\in M_{12}$, then $C_i^{\mathrm U}(s_0)=\min\{|x_i-1|,|x_i-s_0|+s_0\}$. If $x_i\ge s_0$, then $C_i^{\mathrm U}(s_0)\le\min\{|x_i-1|,|x_i|\}\le\OPT_{\MC}(\mathbf c)$. If $x_i<s_0$, then by Lemma~\ref{lem:minmax-control-lb}, $\OPT_{\MC}(\mathbf c)\ge s_0$. The proof of Lemma~\ref{lem:minmax-control-lb} also shows that $s_0\le1/2$. Hence $x_i<1/2$, and by the lower bound above, $\min\{|x_i-1|,|x_i|\}=|x_i|\le\OPT_{\MC}(\mathbf c)$. Therefore,
\[
C_i^{\mathrm U}(s_0)\le |x_i-s_0|+s_0=2s_0-x_i\le2s_0+|x_i|\le3\OPT_{\MC}(\mathbf c).
\]

The argument for agents on $L_2$ is symmetric. By the lower bounds above and Lemma~\ref{lem:minmax-control-lb}, for every $j\in N_1\cup N_{12}$, we have $C_j^{\mathrm L}(s_0)\le3\OPT_{\MC}(\mathbf c)$.

Therefore every agent has cost at most $3\OPT_{\MC}(\mathbf c)$ at the output $s_0$, and hence $\MC(s_0)\le3\OPT_{\MC}(\mathbf c)$.

\noindent\textbf{Tight Example.}
Consider an instance with two agents. Agent $i\in M_2$ is on $L_1$ at $x_i=-1/2$, and agent $j\in N_{12}$ is on $L_2$ at $y_j=1/2$. Then $A=\{i\}$, $B=\{j\}$, $a=-1/2$, and $b=1/2$. Mechanism~4 outputs $s_0=\operatorname{med}\{0,\tfrac{1}{2},1,a,b\}=\tfrac{1}{2}$. At $s_0=1/2$, $C_i^{\mathrm U}(s_0)=3/2$ and $C_j^{\mathrm L}(s_0)=1/2$, so $\MC(s_0)=3/2$. At $s=0$, both agents have cost $1/2$, so $\OPT_{\MC}(\mathbf c)\le1/2$. Also, for every $s\in[0,1]$, agent $i$'s cost is $C_i^{\mathrm U}(s)=1/2+2s\ge1/2$, and hence $\OPT_{\MC}(\mathbf c)=1/2$. Therefore, $\MC(s_0)/\OPT_{\MC}(\mathbf c)=3$.
\qed
\end{proof}

\subsection*{Positive Traversal Costs}

Suppose that crossing the bridge incurs a fixed public cost
$\tau\geq0$, independent of the bridge position and the agents'
reports. The crossing travel costs become $d_{i2}^U(s)+\tau$
and $d_{j1}^L(s)+\tau$, while the local travel costs remain
unchanged. We consider the mechanisms defined in the main text
with their selection rules unchanged.

\paragraph{Sum-variant.}
For each agent interested in the opposite facility, her true
cost increases by $\tau$ at every bridge position. For every
other agent, her true cost remains unchanged. Since this
additive term is independent of the output, the strategyproofness
of Mechanism~2 in Theorem~2 implies that no agent can benefit
from a misreport.
For any fixed true instance, let $SC^\tau(s)$ denote the
social cost with traversal cost $\tau$. Then
\[
SC^\tau(s)=SC(s)+\tau(|A|+|B|).
\]
The added term is independent of $s$. Hence, Mechanism~2
remains strategyproof and optimal for the social cost.

\paragraph{Max-variant.}
Take any $\tau>0$. Consider an instance with two agents:
an agent $i\in M_{12}$ on $L_1$ at $x_i=0$, and an agent
$j\in N_1$ on $L_2$ at $y_j=1$.
Under Mechanism~1, $a_i=1/2$ and $b_j=1$, so
$K=\{0,1/2,1\}$ and the output is $s=1/2$.
Agent $i$'s true cost is $\max\{1,1+\tau\}=1+\tau$.
If she reports preference $(0,1)$ without changing her
location, her breakpoint becomes $0$. The resulting multiset
is $K'=\{0,0,1\}$, so the output is $s'=0$.
Her true cost is now $\max\{1,\tau\}<1+\tau$.
Therefore, Mechanism~1 is not strategyproof for any $\tau>0$.

\paragraph{Min-variant.}
Consider an agent $i\in M_{12}$ at $x_i=0$.
For $s\in[0,1]$, her cost becomes $\min\{1,2s+\tau\}$.
If $0\leq\tau<1$, the two travel costs are equal at
$s=(1-\tau)/2$. If $\tau\geq1$, the local facility is
always a nearest interested facility.
Thus, a positive traversal cost can change the nearest
interested facility. The strategyproofness and approximation
guarantees of the existing mechanisms under this extension
require further analysis.

\end{document}